\documentclass{article}
\usepackage{amssymb,amsmath,amsthm,amsfonts}
\usepackage{framed,fullpage,nicefrac,graphicx,comment}
\usepackage{tikz}
\usetikzlibrary{decorations,positioning}
\usepackage{cite}

\newcommand{\pl}{\ensuremath{\,\mathrm{\mathbf{pl}}}}

\newcommand{\cC}{\ensuremath{\mathcal{C}}}

\newcommand{\R}{{\mathbb R}}

\newcommand{\F}{{\mathbb F}}

\newcommand{\PR}[1]{{\mathbb{P}}\left\{ #1\right\}}
\newcommand{\E}[1]{\mathbb{E}\left[ #1\right]}
\newcommand{\EE}{\mathbb{E}}

\newcommand{\norm}[1]{\left\|#1\right\|}

\newcommand{\inabs}[1]{\left|#1\right|}

\newcommand{\ind}[1]{\ensuremath{\mathbf{1}_{#1}}}

\newcommand{\agr}{\mathrm{agr}}

\newcommand{\argmax}{\mathrm{argmax}}

\newcommand{\poly}{\mathrm{poly}}

\newcommand{\hf}{\frac{1}{2}}

\newcommand{\ip}[2]{\ensuremath{\left\langle #1,#2\right\rangle}}
\newcommand{\inset}[1]{\left\{#1\right\}}
\newcommand{\inparen}[1]{\left(#1\right)}

\newcommand{\suchthat}{\,:\,}

\newcommand{\eps}{\varepsilon}

\newtheorem{theorem}{Theorem} 
\newtheorem*{unlabeledthm}{Theorem}
\newtheorem{lemma}{Lemma}

\newtheorem{cor}{Corollary} 

\newtheorem{remark}{Remark}
\newtheorem{ques}[remark]{Question}

\newtheorem{proposition}{Proposition}

\newtheorem{fact}[theorem]{Fact}

\newtheorem{notation}{Notation}

\newcommand{\Cthm}{C_0}
\newcommand{\Csmallq}{\inparen{2 + \sqrt{2}}}

\newcommand{\Cgauss}{C_3}
\newcommand{\Clem}{C_4}
\newcommand{\Cprop}{C_5}
\newcommand{\Csize}{C_6}
\newcommand{\Ctm}{C_7}
\newcommand{\Ccmp}{\sqrt{2\pi}}
\newcommand{\thisL}{ {\Lambda_t} }
\newcommand{\thisI}{I_t}

\newcommand{\nextI}{I_{t+1}}
\newcommand{\prevL}{{\Lambda_{t-1}}}
\newcommand{\nextL}{{\Lambda_{t+1}}}
\newcommand{\startI}{I_0}
\newcommand{\startL}{\Lambda_0}
\newcommand{\tm}{t_{\max}}

\newcommand{\rate}[1]{R^*(#1)}

\title{Every list-decodable code for high noise has abundant near-optimal rate puncturings\thanks{AR's research is supported in part by NSF CAREER grant CCF-0844796 and NSF grant CCF-CCF-1161196. MW's research is partially supported by the Simons Institute and a Rackham predoctoral fellowship. MW also thanks the Simons Institute for their hospitality.}}
\date{\today\\
}
\author{Atri Rudra\footnotemark[2] \and Mary Wootters\footnotemark[3]}
\date{\today\\
\vspace*{5mm}
\footnotemark[2] ~~University at Buffalo (SUNY)\\
{\tt atri@buffalo.edu}\\
\vspace*{3mm}
\footnotemark[3]~~University of Michigan\\
{\tt wootters@umich.edu}}
\begin{document}
\maketitle
\begin{abstract}
We show that any $q$-ary code with sufficiently good distance can be randomly punctured to obtain, with high probability, a code that is list decodable up to radius $1 - 1/q - \eps$ with near-optimal rate and list sizes.  

Our results imply that ``most" Reed-Solomon codes are list decodable beyond the Johnson bound, settling the long-standing open question of whether \em any \em Reed Solomon codes meet this criterion.  More precisely, we show that a Reed-Solomon code with random evaluation points is, with high probability, list decodable up to radius $1 - \eps$ with list sizes $O(1/\eps)$ and rate $\widetilde{\Omega}(\eps)$.  
As a second corollary of our argument, we obtain improved bounds on the list decodability of random linear codes over large fields.  

Our approach exploits techniques from high dimensional probability.  Previous work used similar tools to obtain bounds on the list decodability of random linear codes, but the bounds did not scale with the size of the alphabet.  In this paper, we use a chaining argument to deal with large alphabet sizes.
\end{abstract}

\newpage

\section{Introduction}
List decoding, proposed by Elias~\cite{elias} and Wozencraft~\cite{wozencraft},
is a relaxation of the traditional notion of unique decoding.  In this
relaxation, the decoder is allowed to output a small list of potentially
transmitted messages with the guarantee that the transmitted codeword is in the
list. 

A remarkable fact about list decoding is that it effectively
doubles the correctable fraction of errors. For any code over alphabet of size $q$, no
more than a $\frac{1}{2}\inparen{1-\frac{1}{q}}$
fraction of errors can be decoded uniquely.
However, when the decoder may output a short list, there are codes
which can tolerate a 
$1-\frac{1}{q}-\eps$ fraction of errors, for any $\eps > 0$.
This fact has been crucially exploited in numerous
applications of list decoding in theoretical computer science and in
particular, in complexity theory.\footnote{See the survey by Sudan~\cite{madhu-survey} and
Guruswami's thesis~\cite{venkat-thesis} for more on these applications.}
There
are two important features of these applications:
\begin{enumerate}
\item Even though in the
traditional communication setting it makes sense to consider constant fraction
$\rho$ of errors (in particular, $\rho$ is close to $0$), for complexity
applications 
it is necessary for the fraction of correctable errors to be arbitrarily close to
$1-\frac{1}{q}$.
\item The optimal rate to correct
$1-\frac{1}{q}-\eps$ fraction of errors is known, and is given by
\[\rate{q,\eps} := 1 - H_q(1 - 1/q - \eps) = \min\inset{ \eps, \frac{ q\eps^2}{2\log(q)} + O_q(\eps^3)}.\]
However, for complexity applications it is often enough to design a code with rate $\Omega(\rate{q,\eps})$
with the same error correction capability.\footnote{In
fact in some applications even polynomial dependence on $\rate{q,\eps}$ is
sufficient.} 
\end{enumerate}
In this paper, we consider the list decoding problem in these
parameter regimes.  That is, we seek to correct a  $1 - 1/q - \eps$ fraction of errors, 
with rate $\widetilde{\Omega}(\rate{q,\eps})$ which may be suboptimal by multiplicative factors. 
The quest for such codes comes in two flavors: one can ask about the list decodability of a specific family of codes, or one can ask for the most general conditions which guarantee list decodability.  This work addresses open problems of both flavors, discussed more below.

\paragraph{Specific families of codes with near-optimal rate.} Many complexity applications require efficient correction of
$1-1/q-\eps$ fraction of errors, sometimes even
with a local decoding algorithm. Thus, there has been significant effort directed at
designing efficiently-decodable codes with optimal rate. 
The first non-trivial progress
towards this goal was due to work of Sudan~\cite{sudan} and
Guruswami-Sudan~\cite{GS99} who showed that 
\em Reed-Solomon \em (RS) codes\footnote{An RS code encodes a low-degree univariate polynomial $f$ over $\F_q$ 
as a list of evaluations $(f(\alpha_1),\ldots, f(\alpha_n))$ for a predetermined set of $n \leq q$ {\em evaluation points} in $\F_q$.}
can be list decoded efficiently from $1-\eps$
fraction of errors with rate $\eps^2$. This matches the so-called
\textit{Johnson bound}, which relates the fraction of errors any code
can combinatorially list decode (with small list size) to the distance of the
code.

The work of Guruswami and Sudan held the record for seven years, during which
RS codes enjoyed the best known tradeoff between rate and fraction of
correctable errors.  However, Parvaresh and Vardy showed that a variant of
Reed-Solomon codes can beat the Johnson bound~\cite{PV05}. This was then
improved by Guruswami and Rudra who achieved the optimal rate of $\eps$ with
Folded Reed-Solomon codes~\cite{GR08}. Since then this optimal
rate result has been achieved with other codes: derivative codes~\cite{GW13},
multiplicity codes~\cite{K12}, folded Algebraic Geometric (AG)
codes~\cite{GX12} as well as {\em subcodes} of RS and AG codes~\cite{GX13}.
There has also been a lot of recent work on reducing the runtime and list size
for folded RS codes~\cite{GW13,DL12,GK13}.

Even though many of the recent developments on list decoding are based on Reed-Solomon codes, there has been no non-trivial progress on the list decodability of Reed-Solomon codes themselves since the work of Guruswami-Sudan.  This is true even if we only ask for combinatorial (not necessarily efficient) decoding guarantees, and even for rates only slightly beyond the Johnson bound.  
The question of whether or not Reed-Solomon codes can be list decoded beyond the Johnson bound 
was our main  motivation for this work: 
\begin{ques}
\label{q:rs}
Are there Reed-Solomon codes which can be combinatorially list decoded from a $1 - \eps$ fraction of errors, with rate $\omega\inparen{ \eps^2 }$?  
\end{ques}
This question, which has been well-studied, is interesting for several reasons.  First, Reed-Solomon codes
themselves are arguably the most well-studied codes in the literature. Secondly,
there are complexity applications where one needs to be able
to list decode Reed-Solomon codes in particular: e.g. the average-case hardness of
the permanent~\cite{CPS99}. Finally, the Johnson bound is a natural barrier and
it is an interesting 
to ask whether it can be overcome by natural codes.\footnote{We note that it is easy to come up with codes that
have artificially small distance and hence can beat the Johnson bound.} It is
known that Reed-Muller codes (which are generalizations of RS codes) can be
list decoded beyond the Johnson bound~\cite{G10,GKZ08}.

There have been some indications that Reed-Solomon codes might \em not \em be list decodable beyond the Johnson bound. 
Guruswami and Rudra~\cite{GR06} 
showed that for a generalization of list decoding called list recovery, the
Johnson bound indeed gives the correct answer for RS codes. Further, Ben-Sasson et
al.~\cite{BKR10} showed that for RS code where the evaluation set is all of $\F_q$,
the correct answer is close to the Johnson bound.
In particular, they show that to correct
$1-\eps$ fraction of errors with polynomial list sizes, the RS code with $\F_q$ as its evaluation points
cannot have rate better than $\eps^{2-\gamma}$ for any constant $\gamma>0$. 
However, this
result leaves open the possibility that one could choose the evaluation points
carefully and obtain an RS code which can be combinatorially list decoded significantly beyond the Johnson bound. 

Resolving the above possibility has been open since~\cite{GS98}: see e.g.~\cite{venkat-thesis,atri-thesis,salil-book} for explicit formulations of this question. 

\paragraph{Large families of codes with near-optimal rate.} While the work on list decodability of specific families of codes have typically also been accompanied with list decoding algorithms, combinatorial results have tended to focus on larger classes of codes. Two classic results along these lines are (i) that random (linear) codes have optimal rate with high probability, and (ii) the fact, following from the Johnson bound, that any code with distance $1 - 1/q - \eps^2$ can be list decoded from $1 - 1/q - \eps$ fraction of errors.

Results of the second type are attractive since they guarantee list decodability for any code, deterministically, as long as the code has large enough distance.  Unfortunately, it is known that the Johnson bound is tight for some codes~\cite{GS03}, and so we cannot obtain a stronger form of (ii).  However, one can hope for a result of the first type for list decodability, based on distance.  More specifically, it is plausible that most \em puncturings \em of a code with good distance can beat the Johnson bound.

Recently, Wootters~\cite{woot2013} obtained such a result for constant $q$. In particular, that work shows that any code with distance $1-1/q-\eps^2$ has many puncturings of rate $\Omega(\eps^2/\log{q})$ that are list decodable from a $1-1/q-\eps$ fraction of errors. 
This rate is optimal up to constant factors when $q$ is small, but is far from the optimal bound of $\rate{q,\eps}$ for larger values of $q$, even when $q$ depends only on $\eps$ and is otherwise constant.
This leads to our second motivating question, left open from~\cite{woot2013}: 
\begin{ques}
\label{q:largeq}
Is it true that any code with distance $1-1/q-\eps^2$ has many puncturings of rate $\widetilde{\Omega}(\rate{q,\eps})$ that can list decode from $1-1/q-\eps$ fraction of errors?
\end{ques}

\bigskip

\paragraph{Our Results.}
In this work, we answer Questions~\ref{q:rs} and \ref{q:largeq} in the affirmative. 
Our main result addresses Question~\ref{q:largeq}. 
We show that random puncturings of any code with distance $1-1/q-\eps^2$ can list decode from $1-1/q-\eps$ fraction of errors with rate
\[\frac{\min\inset{\eps,q\eps^2}}{\log(q)\log^5(1/\eps)}.\]
This improves upon the best known result in this regime by Wootters~\cite{woot2013} for $q\gtrsim \log^5(1/\eps)$, and is optimal up to polylogarithmic factors.
A corollary of this is that random linear codes are list decodable from $1-1/q-\eps$ fraction of errors with the same rate---this improves the corresponding result in~\cite{woot2013} for the same range of parameters.

Our main result also implies a positive answer to Question~\ref{q:rs}, and we show that there do exist RS codes that are list decodable beyond the Johnson bound. 
 In fact, most sets of evaluation points will work: we show that if an appropriate number of evaluation points are chosen at random, then with constant probability the resulting RS code is list decodable from $1-\eps$ fraction of errors with rate 
\[\frac{\eps}{\log(q)\log^5(1/\eps)}.\] 
This beats the Johnson bound for 
\[\eps\le \widetilde{O}\inparen{ \frac{1}{\log(q)} }.\]

\bigskip

\paragraph{Relationship to impossibility results.} 
Before we get into the details, we digress a bit to explain why our result on
Reed-Solomon codes does not contradict the known impossibility results on this question. The
lower bound of~\cite{GR06} works for list recovery but does not apply to our
results about list decoding.\footnote{Our results can be extended to the
list recovery setting, and the resulting parameters obey the lower bound
of~\cite{GR06}. 
} The
lower bound of~\cite{BKR10} does work for list decoding, but critically needs
the set of evaluation points to be all of $\F_q$ (or more precisely the
evaluation set should contain particularly structured subsets $\F_q$). Since we
pick the evaluation points at random, this property is no longer satisfied.
Finally, Cheng and Wan~\cite{CW07} showed that {\em efficiently} solving the
list decoding problem for RS codes from $1-\eps$ fraction of errors with rate
$\Omega(\eps)$ would imply an efficient algorithm to solve the discrete log
problem.  However, this result does not rule out the list size being small
(which is what our results imply), just that algorithmically computing the list
quickly is unlikely.

\subsection{Approach and Organization}
Our main technical result addresses Question \ref{q:largeq} and states that a randomly punctured
code\footnote{Technically, our construction is slightly different than
randomly punctured codes: see Remark~\ref{rem:sample-vs-puncture}.} 
will retain the list decoding properties of the original code as long as the original code
has good distance.  Our results for RS codes (answering Question~\ref{q:rs}) and random linear codes follow by starting from
the RS code evaluated on all of $\F_q$ and the $q$-ary Hadamard code, respectively.

After a brief overview of terminology in Section \ref{sec:prelim}, we give a
more detailed technical overview of our approach in Section \ref{sec:tech}.  In
Section \ref{sec:mainstatement} we state our main result, Theorem
\ref{thm:mainthm}, about randomly punctured codes, and we apply it to
Reed-Solomon codes and random linear codes.  The remainder of the paper,
Sections \ref{sec:redux} and \ref{sec:gaussian}, are devoted to the proof of
Theorem \ref{thm:mainthm}.
Finally, we conclude with Section~\ref{sec:concl}.

\section{Preliminaries}\label{sec:prelim}
Motivated by Reed-Solomon codes, we consider random ensembles of linear codes over $\F_q$, where the field size $q$ is large. 
A code $\mathcal{C} \subseteq \F_q^n$ is \textbf{linear} if it forms a subspace of $\F_q^n$.  Equivalently, $\mathcal{C} = \inset{ x^T G \suchthat x \in \F_q^k }$ for a \textbf{generator matrix} $G \in \F_q^{k \times n}$.  We refer to $x \in \F_q^k$ as the \textbf{message} and $k$ as the \textbf{message length}.  The length $n$ of the resulting codeword $x^TG$ is called the \textbf{block length}.
 
We will study the list decodability of these codes, up to ``large" error rates $1 - 1/q - \eps$, which is $1 - \Theta(\eps)$ when $q \gtrsim 1/\eps$.  We say that a code $\mathcal{C} \subseteq \F_q^n$ is $(\rho, L)$-\textbf{list decodable} if for all $z \in \F_q^n$, the number of codewords $c \in \mathcal{C}$ with $d(z,c) \leq \rho$ is at most $L$, where $d$ denotes relative Hamming distance.  
We will actually study a slightly stronger notion of list decodability, explicitly studied in \cite{gurnar2013}.  We say that a code $\mathcal{C} \subset \F_q^n$ is $(\rho, L)$-\textbf{average-radius list decodable} if for all $z \in \F_q^n$ and all sets $\Lambda$ of $L+1$ codewords $c \in \mathcal{C}$, the average distance between elements of $\Lambda$ and $z$ is at least $\rho.$
Notice that standard list decoding can be written in this language with the average replaced by a maximum.

In general, one is interested in the trade-off between $\eps$, $L$, and the rate of the code $\mathcal{C}$.  The \textbf{rate} of a linear code $\mathcal{C}$ is defined to be $\dim(\mathcal{C})/n$, where $\dim(\mathcal{C})$ refers to the dimension of $\mathcal{C}$ as a subspace of $\F_q^n$.

We'll consider ensembles of linear codes where the generator vectors are independent; this includes random linear codes and Reed Solomon codes with random evaluation points. 
More precisely, a distribution on the matrices $G$ induces a distribution on linear codes.  We say that such a distribution on linear codes $\mathcal{C}$ has \textbf{independent symbols} if the columns of the generator matrix $G$ are selected independently.

We will be especially interested in codes with randomly sampled symbols, where a new code (with a shorter block length) is created from an old code by including a few symbols of the codeword at random.
Formally, suppose that $\mathcal{C}'$ is a linear code over $\F_q$ with generator matrix $G' \in \F_q^{k \times n'}$.  
Form a new generator matrix $G \in \F_q^{k \times n}$ whose columns are $n$ columns of $G'$ chosen independently at random (possibly with replacement).  We say that the resulting random linear code $\mathcal{C}$ with generator matrix $G$ is a \textbf{randomly sampled} version of $\mathcal{C}'$, with block length $n$.
Notice that randomly sampled codes have independent symbols by definition.
\begin{remark}[Sampling vs. Puncturing]
\label{rem:sample-vs-puncture}
We note that the operation of randomly \em sampling \em a code (a term we just made up) is very similar to that of randomly \em puncturing \em a code (a term with a long and illustrious history).  The only difference is that we sample with replacement, while a randomly punctured code can be viewed as a code where the sampling is done without replacement.   Our method of sampling is convenient for our analysis because of the independence.  However, for the parameter regimes we will work in, collisions are overwhelmingly unlikely, and the distribution on randomly sampled codes is indeed very similar to that of randomly punctured codes.   
\end{remark}

\subsection{Notation}
Throughout, we will consider linear codes $\mathcal{C} \subseteq \F_q^n$ of block length $n$ and message length $k$, with generator matrices $G \in \F_q^{k \times n}$. 
The size of $\mathcal{C}$ will be $|\mathcal{C}| = N$. 
For a message $x \in \F_q^k$, we will write $c = c(x)$  for the encoding $c(x) = x^T G$.
We will be interested in subsets $\Lambda \subseteq \F_q^k$ of size $L$ (the \em list size\em), which we will identify, when convenient, with the corresponding subset of $\mathcal{C}$.

For $x,y \in \F_q^n$, let $\agr(x,y) = n(1 - d(x,y))$ be the number of symbols in which $x$ and $y$ agree. 
 We will use $f(x) \lesssim g(x)$ (or $f(x) \gtrsim g(x)$) to indicate that there is some constant $C$ so that $f(x) \leq Cg(x)$ (resp. $g(x) \leq Cf(x)$) for all $x$.  Throughout, $C_0,C_1, \ldots $ and $c_0,c_1,\ldots$ will denote numerical constants.  
For clarity, we have made no attempt to optimize the constants.
For a vector $v = (v_1,v_2,\ldots,v_n) \in \R^n$ and a set $S \subseteq [n]$, we will use $v_S$ to denote the restriction of $v$ to the coordinates indexed by $S$.  We will use the $\ell_p$ norm $\|v\|_p = \inparen{ \sum_{i=1}^n v_i^p }^{1/p}$, and the $\ell_\infty$ norm $\|v\|_\infty = \max_{j \in [n]} |v_j|.$
We use $\log$ to denote the logarithm base $2$,  and $\ln$ to denote the natural log.

We will also use some machinery about Gaussian processes, but we have made an effort to keep this self-contained.  For the reader's convenience, a few useful facts about Gaussian random variables are recorded in Appendix \ref{app:gauss}. 
Finally, we will also use the following form of Chernoff(-Hoeffding) bound:
\begin{theorem}\label{thm:chernoff}
Let $X_1,\dots,X_m$ be $m$ independent random variables such that for every $i\in [m]$, $X_i\in [a_i,b_i]$, then for the random variable
\[S=\sum_{i=1}^m X_i,\]
and any positive $v\ge 0$, we have
\[\PR{|S-\E{S}| \ge v} \le 2\exp\left(-\frac{2v^2}{\sum_{i=1}^m (b_i-a_i)^2}\right).\]
\end{theorem}

\section{Technical overview}\label{sec:tech}
\label{sec:overview}
In this section, we give a technical overview of our argument, and point out where
it differs from previous approaches.  The most similar argument in the literature is in
\cite{woot2013}, which applies to random linear codes (but \em not \em Reed-Solomon codes).
Below, we point out how our approach deviates, and where our improvements come from.

We first recall the classic proof of list decodability of general random codes. 
For a general random code, a Chernoff bound establishes that for a given $\Lambda$ and $z$, there is only a very small probability that the codewords corresponding to $\Lambda$ are all close to $z$.  This probability is small enough to allow for a union bound over the $q^n \cdot {N \choose L}$ choices for $\Lambda$ and $z$.  However, this argument crucially exploits the independence between the encodings of distinct messages.  
If we begin with a random linear code (or a Reed-Solomon code with random
evaluation points), then codewords are no longer independent, and the above
argument fails.  The classic way around this is to consider only the linearly
independent messages in $\Lambda$; however, this results in exponentially large
list sizes of $q^{\Omega(1/\eps)}$.  The exponential dependence on $\eps$ can
be removed for a \em constant \em fraction of errors, by a careful analysis of
the dependence between codewords corresponding to linearly dependent
messages~\cite{GHK11}.  However, such techniques do not seem to work in the
large-error regime that we consider.

In contrast, the approaches of~\cite{cgv2012,woot2013} avoid analyzing the
dependence between codewords by using tools from high dimensional probability.
These arguments, which imply list decodability results for random linear codes, 
work when the error rate approaches $1 - 1/q$, and they
(implicitly) use an improved union bound to avoid having to union bound over
all $\Lambda$ and $z$.  However, these arguments do not scale well with $q$,
which is crucial for the application to Reed-Solomon codes.  In this work, we
follow the approach of~\cite{woot2013} and use techniques from high dimensional
probability and Gaussian processes to avoid the naive union bound.  However,
our arguments \em will \em scale with $q$, and thus are applicable to Reed-Solomon codes.

Following the approach of~\cite{woot2013}, our proof actually establishes \em average-radius \em list decodability.  
The standard definition of list decodability has to do with bounding the maximum distance of a set $\Lambda \subseteq \mathcal{C}$ of $L$ codewords from its centroid $z \in \F_q^n$.  In contrast, average-radius list decodability is a stronger notion which focuses on the average distance from $\Lambda$ to $z$.  

The advantage of considering average-radius list decoding is that it linearizes the problem; after some rearranging (which is encapsulated in Proposition \ref{prop:dec}), it becomes sufficient to control
\[ \sum_{c \in \Lambda} \agr(z,c) = \sum_{c \in \Lambda} \sum_{j=1}^n \ind{c_j=z_j} \]
uniformly 
over all $\Lambda \subseteq \mathcal{C}$ and all $z \in \F_q^n$.  We will show that this is true in expectation; that is, we will bound
\begin{equation}\label{eq:theplan}
 \EE \max_{\Lambda, z} \sum_{c \in \Lambda} \sum_{j=1}^n \ind{c_j = z_j}. 
\end{equation}
The proof proceeds in two steps.

The first (more straightforward) step is to argue that if the expectation and the maximum over $\Lambda$ were reversed in \eqref{eq:theplan}, then we would have the control we need.
To that end, we introduce a parameter 
\[\mathcal{E} = \max_{|\Lambda| = L} \EE \max_{z \in \F_q^n} \sum_{c \in \Lambda} \sum_{j=1}^n \ind{c_j=z_j}.\]
It is not hard to see that the received word $z$ which maximizes the agreement is the one which, for each $j$, agrees with the plurality of the $c_j$ for $c \in \Lambda$.  That is, 
\[ \max_{z \in \F_q^n} \sum_{c \in \Lambda} \sum_{j=1}^n \ind{c_j=z_j} = \sum_{j=1}^n \max_{\alpha \in \F_q} \inabs{ \inset{ c\in \Lambda \suchthat c_j = \alpha } } =: \sum_{j=1}^n \text{plurality}_j\inparen{\Lambda  }.\]
Thus, to control $\mathcal{E}$, we must understand the expected pluralities.  For our applications, this follows from standard Johnson-bound type arguments.  

Of course, it is generally not okay to switch expectations and maxima; we must also argue that the quantity inside the maximum does not deviate too much from its mean in the worst case.  This is the second and more complicated step of our argument.
We must control the deviation
\begin{equation}\label{eq:pluralitysum}
 \sum_{j=1}^n \inparen{ \text{plurality}_j(\Lambda) - \EE \text{plurality}_j(\Lambda) } 
\end{equation}
uniformly over all $\Lambda$ of size $L$.  
By the assumption of independent symbols (that is, independently chosen evaluation points for the Reed-Solomon code, or independent generator vectors for random linear codes), each summand in \eqref{eq:pluralitysum} is independent.  

Sums of independent random variables tend to be reasonably concentrated, but, as pointed out above, because the codewords are not independent there is no reason that the pluralities themselves need to be particularly well-concentrated.  Thus, we cannot handle a union bound over all $\Lambda \subseteq \mathcal{C}$ of size $L$.
Instead, we use a \em chaining argument \em to deal with the union bound.
The intuition is that if the set $\Lambda$ is close to the set $\Lambda'$ (say they overlap significantly), then we should not have to union
bound over both of them as though they were unrelated.  

Our main theorem, Theorem \ref{thm:mainthm}, bounds the deviation \eqref{eq:pluralitysum}, and thus bounds \eqref{eq:theplan} in terms of $\mathcal{E}$.  We control $\mathcal{E}$ in the Corollaries \ref{cor:smallq} and \ref{cor:largeq}, and then explain the consequences for Reed-Solomon codes and random linear codes in Sections \ref{ssec:rs} and \ref{ssec:randlin}.

We prove Theorem \ref{thm:mainthm} in Section \ref{sec:redux}.
To carry out the intuition above, we first pass to the language of Gaussian processes.  Through some standard tricks from high dimensional probability, it will suffice to instead bound the Gaussian process
\begin{equation}\label{eq:fakegp}
 X(\Lambda) = \sum_{j=1}^n g_j \text{plurality}_j(\Lambda).
\end{equation}
uniformly over all $\Lambda$ of size $L$, where the $g_j$ are independent standard normal random variables.

So far, this approach is similar to that of~\cite{woot2013}.  The difference is that Wootters first maps the problem to $\mathbb{C}$, using a technique from \cite{cgv2012}, in a way that allows for a slick bound on the relevant Gaussian process.  However, this approach loses information about the size of $q$.  In particular, the expected size of the pluralities decreases as $q$ increases, and the approach of~\cite{woot2013} does not take advantage of this.  In our approach, we deal with the pluralities directly, without embedding into $\mathbb{C}$.  This unfortunately gives up on the slickness (our argument is somewhat technical), but allows us to take advantage of large $q$.  We outline our methods below.

Returning to the Gaussian process \eqref{eq:fakegp}, we condition on $\mathcal{C}$, considering only the randomness over the Gaussians.  
We control this process in Theorem \ref{thm:gaussian}, the proof of which is contained in Section \ref{sec:gaussian}.
The process \eqref{eq:fakegp} induces a metric on the space of sets $\Lambda$: $\Lambda$ is close to $\Lambda'$ if the vectors of their pluralities are close, in $\ell_2$ distance.  Indeed, if $\Lambda$ is close to $\Lambda'$ in this sense, then the corresponding increment $X(\Lambda) - X(\Lambda')$ is small with high probability.  
In this language, the previous intuition about ``wasting" the union bound on close-together $\Lambda$ and $\Lambda'$ can be made precise---for example, Dudley's theorem~\cite{lt,genchain} bounds the supremum of the process in terms of the size of $\eps$-nets with respect to this distance.

Thus, our proof of Theorem \ref{thm:gaussian} boils down to constructing nets
on the space of $\Lambda$'s.  In fact, our nets are quite simple---smaller nets
consist of all of the sets of size $L/2^t$, for $t = 1,\ldots, \log(L)$.
However, showing that the width of these nets is small is trickier.  
Our argument actually uses the structure of the chaining argument that is at the heart
of the proof of Dudley's theorem: instead of arguing that the width of the net
is small, we argue that each successive net cannot have points that are too far
from the previous net, and thus build the ``chain" step-by-step.  
One can of course abtract out a distance argument and apply Dudley's theorem as a black-box.
However, at the point that we are explicitly constructing the chains, we feel that it is more intuitive
to include the entire argument.  To this end, (and to keep the paper self-contained), 
we unwrap Dudley's theorem in Section \ref{sec:prooffromlemma}.

We construct and control our nets in Lemma \ref{lem:chaining}, which we prove
in Section \ref{sec:chaininglemma}.  Briefly, the idea is as follows.  In order
to show that a set $\Lambda$ of size $L/2^t$ is ``close" to some set $\Lambda'$
of size $L/2^{t+1}$, we use the probabilistic method.  We choose a set
$\Lambda' \subseteq \Lambda$ at random, and argue that in expectation (after some
appropriate normalization), the two are ``close." Thus, the desired $\Lambda'$
exists.  However, the expected distance of $\Lambda$ to $\Lambda'$ in fact
depends on the quantity 
\[ Q_t = \max_{|\Lambda| = L/2^t} \sum_{j=1}^n
\text{plurality}_j(\Lambda).\] 
For $t = 0$, this is the quantity that we were
trying to control in the first place in \eqref{eq:theplan}.  Carrying this 
quantity through our argument, we are able to solve for it at the end
and obtain our bound.

Controlling $Q_t$ for $t > 0$ requires a bit of delicacy.  
In particular, as defined above $Q_{\log(L)}$ is deterministically equal to $n$, 
which it turns out is too large for our applications.
To deal with this, we actually chain over not just the $\Lambda$, but also the set of the symbols $j \in [n]$ that we consider.
In fact, if we did not do this trick, we would recover (with some
extra logarithmic factors) the result of~\cite{woot2013} for random linear codes. 

We remark that our argument has a similar flavor to some existing arguments in other domains, for example~\cite{rudelson97,rv08}, where a quantity analogous to $Q_0$ arises, and where analogous nets will work.  Our approach is slightly different (in particular, our proof of distance is structurally quite different), although it is possible that one could re-frame our argument to mimic those. 


\section{Main theorem}\label{sec:mainstatement}
In this section, we state our main technical result, Theorem \ref{thm:mainthm}.
To begin, we first give a slightly stronger sufficient condition for list decodability, called average-radius list decodability (defined above in Section \ref{sec:prelim}).  Average-radius list decodability has been explicitly studied before in~\cite{gurnar2013} and was used in~\cite{woot2013} to prove upper bounds on the list decodability of ensembles of linear codes for constant-sized $q$.  
All of our results will actually show average-radius list decodability, and the following proposition shows that this will imply the standard notion of list decodability.

\begin{proposition}\label{prop:dec}
Suppose that 
\[ \max_{z \in \F_q^n} \max_{\Lambda \subset \F_q^k, |\Lambda| = L} \sum_{x \in \Lambda} \agr(c(x),z) < nL\inparen{\eps + \frac{1}{q}}.\]
Then $\cC$ is $(1 - 1/q - \eps, L-1)$-list decodable. 
\end{proposition}
\begin{proof}
By definition, $\cC$ is $(1 - \nicefrac{1}{q} - \eps, L-1)$-list decodable if for any $z \in \F_q^n$ and any set $\Lambda \subset \F_q^n$ of size $L$, there is at least one message $x \in \Lambda$ so that $\agr(c(x),z)$ is at most $n \inparen{ \eps + \nicefrac{1}{q}}$, that is, if
\[ \max_{z \in \F_q^n} \max_{|\Lambda| = L} \min_{x \in \Lambda} \agr(c(x),z) < n\inparen{ \eps + \frac{1}{q}}.\]
Since the average is always larger than the minimum, it suffices for
\[ \max_{z \in \F_q^n} \max_{|\Lambda| = L} \sum_{x \in \Lambda} \agr(c(x),z) < Ln\inparen{ \eps + \frac{1}{q}},\]
as claimed.
\end{proof}

Our main theorem gives conditions on ensembles of linear codes under which $\EE \max_{z,\Lambda} \sum_{x \in \Lambda} \agr(c(x),z)$ is bounded.
Thus, it gives conditions under which Proposition \ref{prop:dec} holds.
\begin{theorem}\label{thm:mainthm}
Fix $\eps > 0$.
Let $\mathcal{C}$ be a random linear code with independent symbols.  Let
\[ \mathcal{E} = \max_{\Lambda \subset \F_q^k, |\Lambda| = L} \EE_{\mathcal{C}} \max_{z\in \F_q^k} \inparen{ \sum_{x \in \Lambda} \agr(c(x),z) }.\]
Then
\[ \EE_\mathcal{C} \max_{z\in \F_q^n} \max_{\Lambda \subset \F_q^k, |\Lambda| = L} \sum_{x \in \Lambda} \agr(c(x),z) \leq 
\mathcal{E} + Y + \sqrt{ \mathcal{E} Y },
\]
where
\[ Y = \Cthm L \log(N) \log^5(L)\]
for an absolute constant $\Cthm$.
\end{theorem}
Together with Proposition \ref{prop:dec}, Theorem \ref{thm:mainthm} implies results about the list decodability of random linear codes with independent symbols, which we present next. 

\begin{remark}We have chosen the statement of the theorem which gives the best bounds for Reed-Solomon codes, where $q \gg L$ is a reasonable parameter regime.  An inspection of the proof shows that we may replace one $\log(L)$ factor with $\min\{ \log(L),\log(q) \}$.
\end{remark}

\subsection{Consequences of Theorem \ref{thm:mainthm}: list decodability of Reed-Solomon codes and random linear codes}\label{sec:consequence}
In this section, we derive some consequences of Theorem \ref{thm:mainthm} for randomly sampled codes, in terms of the distance of the original code. 
Our motivating examples are Reed-Solomon codes with random evaluation points, and random linear codes, which both fit into this framework.
Indeed, Reed-Solomon codes with random evaluation points are obtained by sampling symbols from the Reed-Solomon code with block length $n=q$, and a random linear code is a randomly sampled Hadamard code.
We'll discuss the implications
and optimality for the two motivating examples below in Sections \ref{ssec:rs} and \ref{ssec:randlin} respectively.

Our corollaries are split into two cases: 
the first holds for all $q$, but only yields the correct list size when $q$ is small.  The second holds for $q \gtrsim 1/\eps^2$, and gives an improved list size in this regime.
As discussed below in Section \ref{ssec:randlin}, our results are nearly optimal in both regimes.

First, we prove a result for intended for use with small $q$.
\begin{cor}[Small $q$]\label{cor:smallq}
Let $\mathcal{C}'$ be a linear code over $\F_q$ with distance $1 -\frac{1}{q}-\frac{\eps^2}{2}$.
Suppose that
\[ n \geq \frac{ \Cthm \log(N) \log^5(L) }{\min \inset{ \eps, q\eps^2 }},\]
and choose $\mathcal{C}$ to be a randomly sampled version of $\mathcal{C}'$, of block length $n$.
Then, with constant probability over the choice of $\mathcal{C}$,
the code $\mathcal{C}$ is $(1 - \nicefrac{1}{q} - \eps', 2/\eps^2)$-list decodable,
where $\eps' = \Csmallq \eps$.
\end{cor}

Corollary \ref{cor:smallq} holds for all values of $q$, but the list size $L \gtrsim \eps^{-2}$ is suboptimal when $q \gtrsim 1/\eps$.
To that end, we include the following corollary, which holds when $q \gtrsim 1/\eps^2$ and attains the ``correct'' list size.\footnote{As discussed below, we do not know good lower bounds on list sizes for large $q$; by ``correct" we mean matching the performance of a general random code.}

\begin{cor}[Large $q$]\label{cor:largeq}
Suppose that $q > 1/\eps^2$, and that $\eps$ is sufficiently small.
Let $\mathcal{C}'$ be a linear code over $\F_q$ with distance $1 -\eps^2$.
Let
\[ n \geq \frac{ 2\Cthm \log(N) \log^5(L) }{\eps},\]
and choose $\mathcal{C}$ to be a randomly sampled version of $\mathcal{C}'$, of block length $n$.
Then, with constant probability over the choice of $\mathcal{C}$,
the code $\mathcal{C}$ is $(1 - \eps', 1/\eps)$-list decodable,
where $\eps' = 5\eps$.
\end{cor}

The proofs of Corollaries \ref{cor:smallq} and \ref{cor:largeq} amount to controlling the worst expectation $\mathcal{E}$.  
This control follows from standard Johnson bound-type statements, and the proofs are given in Appendix \ref{app:avgjb}.
Below, we discuss the consequences (and optimality) of these corollaries for Reed-Solomon codes and random linear codes.

\begin{remark}[Average-radius list decodability]
We remark that the proofs of both Corollaries \ref{cor:smallq} and \ref{cor:largeq} go through Proposition \ref{prop:dec}, and thus actually show average-radius list decodability, not just list decodability.  In particular, the applications to both Reed-Solomon codes and random linear codes hold under this stronger notion as well.
\end{remark}

\subsection{Most Reed-Solomon codes are list-decodable beyond the Johnson bound}\label{ssec:rs}
Our results imply that a Reed-Solomon code with random evaluation points is, with high probability, list decodable beyond the Johnson bound.  

We briefly recall the definition of Reed-Solomon codes, and set notation for our discussion.
Fix $q \geq n$, and an integer $k$, and let $\inset{\alpha_1,\ldots,\alpha_n} \subseteq \F_q$ be a list of ``evaluation points."
The corresponding \textbf{Reed-Solomon code}  $\mathcal{C} \subset \F_q^n$ encodes a polynomial (message) $f \in \F_q[x]$ of degree at most $k-1$ as
\[ c(f) = ( f(\alpha_1), f(\alpha_2), \ldots, f(\alpha_n) ) \in \F_q^n. \]
Note that there are $q^k$ polynomials of degree at most $k-1$, and thus $|\mathcal{C}| = q^k$.

For Reed-Solomon codes, we are often interested in the parameter regime when $q \ge n$ is quite large.
In particular, below we will be especially interested in the regime when $q \gg 1/\eps^2$, and so we will use Corollary \ref{cor:largeq} for this application.
To apply Corollary \ref{cor:largeq}, let $\mathcal{C}'$ be the Reed-Solomon code of block length $q$ (that is, every point in $\F_q$ is evaluated), and choose the $n$ evaluation points $(\alpha_i)_{i=1}^n$ for $\mathcal{C}$ independently from $\F_q$.
We will choose the block length $n$ so that 
\[ n \lesssim \frac{ \log(N) \log^5(1/\eps) }{\eps}.\]
It is well known that 
the generator matrix for $\mathcal{C}$ will have full rank.  In the favorable case, the rate of $\mathcal{C}$ is at least
\begin{equation}\label{eq:rsrate}
R \gtrsim \frac{ \eps }{\log(q)\log^5(1/\eps) }.
\end{equation}

Before we investigate the result of Corollary \ref{cor:largeq}, let us pause to
observe what the Johnson bound predicts for $\mathcal{C}$.  The distance of
$\mathcal{C}$ is exactly $1 - (k-1)/n$.  Indeed, any two polynomials of degree
$k-1$ agree on at most $k-1$ points, and this is attained by, say, the zero
polynomial and any polynomial with $k$ distinct roots in $\{\alpha_1,\ldots,
\alpha_n\}$.  Thus, letting $\eps = (k-1)/n$, the Johnson bound predicts that
$\mathcal{C}$ has rate $\eps$, distance $1 - \eps$, and is list decodable up to
$1 - O(\sqrt{\eps})$, with polynomial list sizes.

Now, we compare this to the result of Corollary \ref{cor:largeq}.
The distance of $\mathcal{C}'$ is $1 - (k-1)/q$, so as long as $q \gtrsim k/\eps^2$, we may apply Corollary \ref{cor:largeq}.  
Then, Corollary \ref{cor:largeq} implies that the resulting Reed-Solomon code $\mathcal{C}$ has rate 
\[ \Omega\inparen{ \frac{\eps}{\log(q)\log^5(1/\eps)} }, \]
distance $1 - \eps$, and is list decodable up to radius $1 - 5\eps$, with list sizes at most $1/\eps$.

In particular, the tolerable error rate may be as large as $1 - O(\eps)$, rather than $1 - O(\sqrt{\eps})$, and the rate suffers only by logarithmic factors.

\subsection{Near-optimal bounds for random linear codes over large alphabets}\label{ssec:randlin}
In addition to implying that most Reed-Solomon codes are list decodable beyond the Johnson bound, Corollaries \ref{cor:smallq} and \ref{cor:largeq} provide the best known bounds on random linear codes over large fields.  This improves the recent work of one of the authors in~\cite{woot2013} for large $q$; further, our new results are tight up to logarithmic factors. 

Suppose that $\mathcal{C}'$ is the Hadamard code over $\F_q$ of dimension
$k$; that is, the generator matrix of $\mathcal{C}' \in \F_q^{k \times q^k}$
has all the elements of $\F_q^k$ as its columns.  The relative distance of
$\mathcal{C}'$ is $1 - \nicefrac{1}{q}$, and so we may
apply the corollaries with any $\eps > 0$ that we choose. 

To this end, fix $\eps > 0$, and let $\mathcal{C}$ be a randomly sampled version of $\mathcal{C}'$, of block length
\[ n = \frac{ 2\Cthm \log(q^k) \log^5(1/\eps) }{\eps}. \]
It is not hard to see that the generator matrix of $\mathcal{C}$ will have full rank with high probability,
and so the rate of $\mathcal{C}$ will be at least
\begin{equation}\label{eq:rate}
R = k/n = \frac{ \min\inset{ \eps, q\eps^2 } }{ 2\Cthm \log(q) \log^5(1/\eps) }. 
\end{equation}
By Corollary \ref{cor:smallq}, $\mathcal{C}$ is list decodable up to error radius $1 - \nicefrac{1}{q} - O(\eps)$, with list sizes at most $2/\eps^2$.
When $q \gtrsim 1/\eps^2$, Corollary \ref{cor:largeq} applies, and we get the same result with an improved list size of $1/\eps$.

We compare these results to known results on random linear codes in Figure \ref{table:randlin}.
\begin{figure}
\begin{tikzpicture}[node distance=0]
\tikzstyle{entry}=[rectangle,minimum width=1.4in,text width=1in, align=center,minimum height=1in]
\tikzstyle{axis}=[rectangle,minimum width=1.4in, text width=1in, align=center,minimum height=.6in]
\node[axis](rate) {Best known rate for random linear codes};
\node[axis,right=of rate](upper) {Upper bound on rate};
\node[axis,right=of upper](list) {Best known list size for random linear codes};
\node[axis,right=of list](lower) {Lower bound on list size};

\node[entry,below=of rate,minimum height=.6in](woot){$\frac{ \eps^2}{\log(q)}$,\cite{woot2013}};
\node[entry,below=of woot](ussmall){\[\frac{ q \eps^2}{\log(q) \log^5(1/\eps)}\] this work \\ Cor. \ref{cor:smallq} };
\node[entry,below=of ussmall,minimum height=1.4in](uslarge){ \[\frac{ \eps}{\log(q) \log^5(1/\eps)}\]  this work \\ Cors. \ref{cor:smallq}, \ref{cor:largeq}};

\node[entry,below=of upper, minimum height=1.6in](optimum){\[\frac{ q \eps^2 }{ \log(q)}\]};
\node[entry,below=of optimum, minimum height=1in](opt2){\[1-H_q\inparen{1 - \frac{1}{q} - \eps}\]};
\node[entry,below=of opt2, minimum height=.4in](opt3){$\eps$};

\node[entry,below=of list,minimum height=2in](easy){\[\frac{1}{ \eps^2}\] \cite{cgv2012,woot2013}, this work \\ Cor. \ref{cor:smallq}};
\node[entry,below=of easy,minimum height=1in](eps){\[\frac{1}{ \eps}\]  this work \\ Cor. \ref{cor:largeq}};

\node[entry,below=of lower,minimum height=2in](gv){\[\frac{1}{ q^5 \eps^2}\] \cite{GV10}};
\node[entry,below=of gv, minimum height=1in](filler) {};

\draw (woot.north east) -- (uslarge.south east);
\draw (optimum.north east) -- (opt3.south east);
\draw (easy.north east) -- (eps.south east);
\draw[very thick] (gv.north east) -- (filler.south east);
\draw[very thick] (woot.north west) -- (uslarge.south west);
\draw[very thick] (rate.south west) -- (lower.south east);
\draw[very thick] (uslarge.south west) -- (filler.south east);

\draw[dashed] (woot.south east) -- (woot.south west);
\draw[dashed] (optimum.south east) -- (ussmall.south west);
\draw[dashed] (opt2.south east) -- (opt2.south west);
\draw[dashed] (easy.south west) -- (gv.south east);

\node[left=.1in of woot.south west](t1) {$q = \log^5(1/\eps)$};
\node[left=.1in of ussmall.south west](t2) {$q = 1/\eps$};
\node[left=1.5in of opt2.south west](t3) {$q = 2^{\Omega(1/\eps)}$};
\node[left=4.3in of gv.south west](t4) {$q = 1/\eps^2$};
\node[left= .1in of woot.north west](t4) {Regime};

\end{tikzpicture}
\caption{The state of affairs for $q$-ary random linear codes.  Above, the list
decoding radius is $1 - 1/q - \eps$, and we have suppressed constant factors.}
\label{table:randlin}
\end{figure}
The best known results on the list decodability of random linear codes, from ~\cite{woot2013}, state that a random linear code of rate on the order of $\eps^2 / \log(q)$ is $(1 - 1/q - \eps, O(1/\eps^2))$-list decodable.  This is optimal (up to constant factors) for constant $q$, but it is suboptimal for large $q$.  In particular, the bound on the rate is surpassed by our bound \eqref{eq:rate} when $q \gtrsim \log^5(1/\eps)$.

When the error rate is $1 - 1/q - \eps$, the optimal information rate for list decodable codes is given by the list decoding capacity theorem, which implies that we must have $R \leq 1 - H_q(1 - 1/q - \eps).$
This expression behaves differently for different parameter regimes; in particular, when $q \leq 1/\eps$ and $\eps$ is sufficiently small, we have
\[ 1 - H_q(1 - 1/q - \eps) = \frac{ q\eps^2 }{2 \log(q) (1 - 1/q) } + O(\eps^3), \]
while when $q \geq 2^{\Omega(1/\eps)}$, the optimal rate is linear in $\eps$.
For the first of these two regimes---and indeed whenever $q \leq 1/\poly(\eps)$---our bound \eqref{eq:rate} is optimal up to polylogarithmic factors in $1/\eps$.  In the second regime, when $q$ is exponentially large, our bound slips by an additional factor of $\log(q)$.  

For the $q \leq 1/\eps^2$ regime, our list size of $1/\eps^2$ matches existing results, and when $q$ is constant it matches the lower bounds of~\cite{GV10}.
For $q \geq 1/\eps^2$, our list size of $1/\eps$ is the best known. 
There is a large gap between the lower bound of~\cite{GV10} and our upper bounds for large $q$.  However, there is evidence that the most of discrepancy is due to the difficulty of obtaining lower bounds on list sizes.  Indeed, a (general) random code of rate $1 - H_q(1-1/q-\eps) - 1/L$ is list-decodable with list size $L$, implying that $L = O(1/\eps)$ is the correct answer for $q \gtrsim 1/\eps$.  
Thus, while our bound seems like it is probably weak for $q$ super-constant but smaller than $1/\eps^2$, it seems correct for $q \gtrsim 1/\eps^2$.

\section{Proof of Theorem \ref{thm:mainthm}: reduction to Gaussian processes}\label{sec:redux}

In this section, we prove Theorem \ref{thm:mainthm}.  
For the reader's convenience, we restate the theorem here.
\begin{unlabeledthm}[Theorem \ref{thm:mainthm}, restated]
Fix $\eps > 0$.
Let $\mathcal{C}$ be a random linear code with independent symbols.  Let
\[ \mathcal{E} = \max_{\Lambda \subset \F_q^k, |\Lambda| = L} \EE_{\mathcal{C}} \max_{z\in \F_q^k} \inparen{ \sum_{x \in \Lambda} \agr(c(x),z) }.\]
Then
\[ \EE_\mathcal{C} \max_{z\in \F_q^n} \max_{\Lambda \subset \F_q^k, |\Lambda| = L} \sum_{x \in \Lambda} \agr(c(x),z) \leq 
\mathcal{E} + Y + \sqrt{ \mathcal{E} Y },
\]
where
\[ Y = \Cthm L \log(N) \log^5(L)\]
for an absolute constant $\Cthm$.
\end{unlabeledthm}

To begin, we introduce some notation.  
\begin{notation}
For a set $\Lambda \subseteq \F_q^k$, let $\pl_j$ denote the (fractional) \em plurality \em of index $j\in [n]$:
\[ \pl_j(\Lambda) = \frac{1}{|\Lambda|}\max_{ \alpha \in \F_q}  \inabs{ \inset{ x \in \Lambda \suchthat c(x)_j = \alpha } }.\]
For a set $I \subseteq [n]$, let 
\[ \pl_I(\Lambda) \in [0,1]^n \]
be the the vector $(\pl_j(\Lambda))_{j =1}^n$ restricted to the coordinates in $I$, with the remaining coordinates set to zero.
\end{notation}

Rephrasing the goal in terms of our new notation, the quantity we wish to bound is
\begin{equation}\label{eq:want}
\EE_{\cC} \max_{z \in \F_q^n} \max_{ |\Lambda| = L} \sum_{x \in \Lambda} \agr(c(x),z)
=
L \cdot \EE_\mathcal{C} \max_{|\Lambda| = L} \sum_{j \in [n]} \pl_j(\Lambda) .
\end{equation}

Moving the expectation inside the maximum recovers the quantity
\[ \mathcal{E} =  L\cdot \max_{|\Lambda| = L} \EE_\cC \sum_{j \in [n]} \pl_j(\Lambda),\]
which appears in
the statement of Theorem \ref{thm:mainthm}.
Since Theorem \ref{thm:mainthm} outsources 
a bound on $\mathcal{E}$
to the user (in our case, Corollaries \ref{cor:smallq} and \ref{cor:largeq}),
we seek to control the worst deviation
\begin{align}
\mathcal{F} &:= L \cdot \EE_\cC \max_{|\Lambda| = L}  \inabs{ \sum_{j \in [n]} \pl_j(\Lambda) - \EE_\cC \sum_{j \in [n]} \pl_j(\Lambda) } \notag\\
&= L\cdot\EE_\cC \max_{|\Lambda| = L} \inabs{ \sum_{j \in [n]} \inparen{\pl_j(\Lambda) - \EE_\cC \pl_j(\Lambda)}  } \label{eq:deviation}.
\end{align}
Indeed, let 
\[Q = Q(\mathcal{C}) = \max_{|\Lambda| = L} \sum_{j\in [n]} \pl_j(\Lambda),\]
so that 
$L\cdot \EE_\mathcal{C} Q$ is the quantity in \eqref{eq:want}. 
Then, 
\begin{align}
\EE_\mathcal{C} Q &= \EE_{\mathcal{C}} \max_{|\Lambda| = L}  \inparen{ \sum_{j \in [n]} \pl_j(\Lambda) - \EE_{\mathcal{C}} \sum_{j \in I} \pl_j(\Lambda)  + \EE_{\mathcal{C}} \sum_{j \in I} \pl_j(\Lambda) }\notag\\
&\leq \EE_{\mathcal{C}} \max_{|\Lambda| = L} \inabs{ \sum_{j \in [n]} \pl_j(\Lambda) - \EE_{\mathcal{C}} \sum_{j \in I} \pl_j(\Lambda)} + \max_{|\Lambda| = L} \EE_{\mathcal{C}} \sum_{j \in [n]} \pl_j(\Lambda) \notag\\
&= \frac{1}{L} \inparen{\mathcal{F} + \mathcal{E}}, \label{eq:triangleineq}
\end{align}
so getting a handle on $\mathcal{F}$ would be enough.  With that in mind, we return our attention to \eqref{eq:deviation}.
By the assumption of independent symbols, the summands in \eqref{eq:deviation} are independent.  
By a standard symmetrization argument followed by a comparison argument (made precise in Appendix \ref{app:justsymm} as Lemmas \ref{lem:symmetrize} and \ref{lem:compare}, respectively), we may bound
\begin{align}\label{eq:symmetrize}
\frac{1}{L} \mathcal{F} &= \EE_\cC \max_{|\Lambda| = L} \inabs{ \sum_{j \in [n]} \inparen{\pl_j(\Lambda) - \EE_\cC \pl_j(\Lambda)}  }\\
&\leq \Ccmp\, \EE_\cC \EE_g \max_{|\Lambda| = L} \inabs{ \sum_{j\in [n]} g_j \pl_j(\Lambda) }
\end{align}
Above, $g_j$ are independent standard normal random variables.

Let
\begin{equation}\label{eq:defS0}
 \mathcal{S}_0 = \{[n]\} \times \inset{ \Lambda \subset \F_q^k \suchthat |\Lambda| = L },
\end{equation}
so that we wish to control
\[ \EE_\cC \EE_g \max_{(I,\Lambda) \in \mathcal{S}_0} \inabs{ \sum_{j \in I} g_j \pl_j(\Lambda) }. \]
At this stage, maximimizing $I$ over the one-element collection $\{[n]\}$ may seem like a silly use of notation, but we will use the flexibility as the argument progresses.

Condition on the choice of $\mathcal{C}$ until further notice, and consider only the randomness over the Gaussian random vector $g = (g_1, \ldots, g_n)$.
In particular, this fixes $Q = Q(\mathcal{C})$.
In order to take advantage of \eqref{eq:symmetrize}, we will study the Gaussian process
\begin{equation}\label{eq:gp}
X(I, \Lambda) = \sum_{j\in I} g_j \pl_j(\Lambda) 
\end{equation}
indexed by $(I,\Lambda) \in \mathcal{S}_0$.
The bulk of the proof of Theorem \ref{thm:mainthm} is the following theorem, which controls the expected supremum of $X(I,\Lambda)$, in terms of $Q$.
\begin{theorem}\label{thm:gaussian}
Condition on the choice of $\cC$.  Then
\begin{align*}
        \EE_g \max_{(I,\Lambda) \in \mathcal{S}_0} |X(I,\Lambda)| &\leq \Cgauss 
                 \sqrt{Q\log(N)\log^5(L)}
\end{align*}
for some constant $\Cgauss$.
\end{theorem}
We will prove Theorem \ref{thm:gaussian} in Section \ref{sec:gaussian}.  First, let us show how it implies Theorem \ref{thm:mainthm}.
By \eqref{eq:symmetrize}, and applying Theorem \ref{thm:gaussian}, we have
\begin{align*}
 \mathcal{F} &\leq \Ccmp \, L \,\EE_\cC \EE_g \max_{(I, \Lambda) \in \mathcal{S}_0} \inabs{\sum_{j\in I} g_j v_j(z,\Lambda)} \\
		&\leq \Cgauss \Ccmp \, L \, \EE_\cC  \left[
                 \sqrt{Q\log(N)\log^5(L)}  \right]\\
		&\leq \Cgauss \Ccmp \, L \,  
                 \sqrt{\EE_{\cC}Q\,\log(N)\log^5(L)} 
\end{align*}
Using the fact \eqref{eq:triangleineq} that $\EE_\cC Q \leq \frac{1}{L} \inparen{\mathcal{E} + \mathcal{F}}$,
\begin{align*}
	\mathcal{F} &\leq \Cgauss\Ccmp \sqrt{ L\inparen{ \mathcal{E} + \mathcal{F}} \log(N) \log^5(L) } \\
			&=: \sqrt{ Y( \mathcal{E} + \mathcal{F} ) }, 
\end{align*}
where 
\[ Y := \Cgauss^2 2\pi L \log(N) \log^5(L). \]
Solving for $\mathcal{F}$, this implies that
\[ \mathcal{F} \leq \frac{Y + \sqrt{ Y^2 + 4Y\mathcal{E} }}{2} \leq Y + \sqrt{Y\mathcal{E}}.\]
Then, from \eqref{eq:triangleineq} and the definition of $Q$ 
(recall that 
$L\cdot \EE_\mathcal{C} Q$ is the quantity in \eqref{eq:want}),
\begin{align*}
\EE_{\cC} \max_{I, \Lambda} \sum_{x \in \Lambda} \agr(c(x),z) &= L\EE_{\mathcal{C}} Q \\
&\leq \mathcal{E} + \mathcal{F} \\
& \leq \mathcal{E} + Y + \sqrt{Y\mathcal{E}} ,
\end{align*}
as claimed.
This proves Theorem \ref{thm:mainthm}.

\section{Proof of Theorem \ref{thm:gaussian}: controlling a Gaussian process}\label{sec:gaussian}
In this section, we prove Theorem \ref{thm:gaussian}.  Recall that the goal was to control the Gaussian process \eqref{eq:gp} given by
\begin{equation*}
	X(I, \Lambda) = \sum_{j \in I} g_j \pl_j(\Lambda). 
\end{equation*}
Recall also that we are conditioning on the choice of $\mathcal{C}$.  Because of this, for notational convenience, we will identify $\Lambda \subset \F_q^k$ with the corresponding set of codewords $\inset{ c(x) \suchthat x \in \Lambda } \subset \mathcal{C}$.  That is, for this section, we will imagine that $\Lambda \subset \mathcal{C}$ is a set of codewords.
\begin{notation}
	When the code $\mathcal{C}$ is fixed (in particular, for the entirety of Section \ref{sec:gaussian}), we will identify $\Lambda \subset \F_q^k$ with $\Lambda \subset \mathcal{C}$, given by
		\[ \Lambda \gets \inset{ c(x) \suchthat x \in \Lambda }. \]
\end{notation}

To control the Gaussian process \eqref{eq:gp}, we will use a so-called ``chaining argument."  
That is, we will define a series of nets, $\mathcal{S}_t \subset 2^{[n]} \times 2^\mathcal{C}$ and write, for any $(\startI,\startL) \in \mathcal{S}_0$,
\begin{equation*}
       \inabs{ X(\startI,\startL) } \leq  \inparen{ \sum_{t=0}^{t_{\max}-1} \left| X( \pi_{t}(\startI,\startL) ) - X(\pi_{t+1}(\startI, \startL) ) \right| } + \inabs{  X( \pi_{t_{\max}} (\startI,\startL) )},
\end{equation*}
where $\pi_t(\startI,\startL) \in \mathcal{S}_t$ will shortly be determined, and $\pi_0(\startI,\startL) = (\startI,\startL)$.  Then we will argue that each step in this ``chain" (that is, each summand in the first term) is small with high probability, and union bound over all possible chains.  

For Gaussian processes, such chaining arguments come in standard packages, for
example Dudley's integral inequality~\cite{lt}, or Talagrand's generic chaining
inequality~\cite{genchain}.   However, we choose to unpack the argument
for two reasons.  The first is that our choice of nets is informed by the
structure of the chaining argument, and so we feel it is clearer to define the nets in 
the context of the complete argument.  
The second reason is to make the exposition self-contained.

We remark that, due to the nature of our argument, it is convenient for us to start with the large nets indexed by small $t$, and the small nets indexed by large $t$; this is in contrast with convention.

\subsection{Defining the nets}
We will define nets $\mathcal{S}_t$, for each $t$ recursively. 
Begin by defining $\mathcal{S}_0$ as in \eqref{eq:defS0}, and let $\pi_0: \mathcal{S}_0 \to \mathcal{S}_0$ be the identity map.
Given $\mathcal{S}_t$, we will define $\mathcal{S}_{t+1}$, as well as the maps $\pi_{t+1}: \mathcal{S}_0 \to \mathcal{S}_{t+1}$.
Our maps $\pi_t$ will satisfy the guarantees of the following lemma.

\begin{lemma}\label{lem:chaining}
Fix a parameter $\eta = 1/\log(L)$, and suppose $c_0 < L < N/2$ is sufficiently large, for some constant $c_0$.
Let
\begin{equation}\label{def:tm}
 \tm = \frac{ \log(L) - 2\log(1/\eta) - 2}{\log(2/(1 - \eta))}.
\end{equation}
Then there is 
a sequence of maps
\[\pi_t : \mathcal{S}_0 \to 2^{[n]} \times 2^{\mathcal{C}} \]
for $t = 0,\ldots, \tm$ so that $\pi_0$ is the identity map and
so that the following hold.

First, for all $(\startI,\startL) \in \mathcal{S}_0$, and for all $t =0,\ldots,\tm$, the pair $(\thisI,\thisL) = \pi_t(\startI, \startL)$ obeys
\begin{equation}\label{eq:induct}
 \sum_{j \in \thisI} \pl_j(\thisL) \leq Q_t := \inparen{ 1 + \eta }^t Q. 
\end{equation}
and
\begin{equation}\label{eq:Lsize}
\inparen{ \frac{1 - \eta}{2}}^{t} L  \leq  |\thisL|  \leq \inparen{ \frac{1 + \eta}{2}}^t L.
\end{equation}

In addition, for all $(\startI,\startL) \in \mathcal{S}_0$, and for all $t = 0,\ldots, \tm -1$, the pair $(\nextI,\nextL) = \pi_{t+1}(\startI,\startL)$ obeys
	\begin{equation}\label{eq:width}
	\norm{ \pl_{\thisI}(\thisL) -  \pl_{\nextI}(\nextL) }_2 \leq  \frac{ \Clem \sqrt{Q_t \log(L)}}{\eta \sqrt{|\thisL|} }
	\end{equation}
	for some constant $\Clem$.

Finally, for all $t = 0,\ldots,\tm$, define
\[ \mathcal{S}_t := \inset{ \pi_t( \startI,\startL) \suchthat (\startI,\startL) \in \mathcal{S}_0 }. \]
Then, for $t \geq 1$, the size of the net $\mathcal{S}_t$ satisfies
\begin{equation}\label{eq:size}
|\mathcal{S}_t| \leq \Csize {N \choose eL/2^t}{N \choose eL/2^{t-1}},
\end{equation}
for some constant $\Csize$,
while $|\mathcal{S}_0| = {N \choose L}$.
\end{lemma}

\subsection{Proof of Theorem \ref{thm:gaussian} from Lemma \ref{lem:chaining}: a chaining argument}\label{sec:prooffromlemma}
Before we prove Lemma \ref{lem:chaining}, we will show how to use it to prove Theorem \ref{thm:gaussian}.
This part of the proof follows the standard proof of Dudley's theorem~\cite{lt}, and can be skipped by the reader already familiar with it.\footnote{Assuming that the reader is willing to take our word on the calculations.}
As outlined above, we will use a chaining argument to control the Gaussian process in Theorem \ref{thm:gaussian}.
We wish to control
\[ \EE \max_{(I,\Lambda) \in \mathcal{S}_0} \inabs{ X(I,\Lambda) }. \]
For any $(\startI,\startL) \in \mathcal{S}_0$,  write
\begin{align}
       \inabs{ X(\startI,\startL) } &\leq  \inparen{ \sum_{t=0}^{t_{\max}-1} \left| X( \pi_{t}(\startI,\startL) ) - X(\pi_{t+1}(\startI, \startL) ) \right| } + \inabs{ X( \pi_{t_{\max}} (\startI,\startL) )}\notag\\
 &=: S(\startI,\startL) + \inabs{ X(\pi_{\tm}(\startI,\startL) ) },\label{eq:chain}
\end{align}
where Lemma \ref{lem:chaining} tells us how to pick $(\thisI, \thisL) := \pi_t(\startI,\startL)$, and where we have used the fact that $\pi_0(\startI,\startL) = (\startI,\startL)$.

Each increment
\[ X(\pi_t(\startI,\startL)) - X(\pi_{t+1}(\startI,\startL)) = \sum_{j=1}^n g_j \left[ \ind{j \in \thisI} \pl_j( \thisL) - \ind{j \in \nextI} \pl_j(\nextL) \right] \]
is a Gaussian random variable (see Fact \ref{fact:2stab} in Appendix \ref{app:gauss}) with variance
\begin{align*}
 \sum_{j=1}^n \inparen{ \ind{j \in \thisI} \pl_j( \thisL) -  \ind{j \in \nextI } \pl_j( \nextL)  }^2 
&= \norm{ \pl_{\thisI}(\thisL) - \pl_{\nextI}(\nextL) }_2^2 \\
&\leq \frac{ \Clem^2  Q_t \log(L) }{\eta^2 |\thisL| } \qquad \text{ by \eqref{eq:width} }\\
&\leq \frac{ \Clem^2 Q_t \log(L) }{\eta^2 \inparen{ \frac{1 - \eta}{2} }^{t} L } \qquad \text{ by \eqref{eq:Lsize}}\\
&\leq \frac{ \Clem^2 (1 + \eta)^t Q \log(L) }{\eta^2 \inparen{ \frac{1 - \eta}{2} }^{t} L } \qquad \text{ by \eqref{eq:induct}}\\
&\leq \inparen{\frac{ \Clem}{\eta} }^2 \inparen{ \frac{ Q \log(L) (2(1+2\eta))^{t} }{L} } \qquad \text{ using $\eta \le 1/2$. }\\
&\leq \inparen{\frac{ e\Clem}{\eta} }^2 \inparen{ \frac{ Q \log(L) 2^{t} }{L} } \qquad \text{ using $\eta = 1/\log(L)$ and $\tm \leq \log(L)$. }
\end{align*}
Thus, for each $0 \leq t < \tm$, and for any $u,a_t \geq 0$,
\begin{align}
	\PR{ | X(\pi_t(z,\Lambda)) - X(\pi_{t+1}(z,\Lambda)) | > u \cdot a_t } &\leq \exp\inparen{ \frac{ -u^2\cdot a_t^2 }{ 2\sum_{j=1}^n \inparen{\ind{j \in \thisI} \pl_j(\thisL) -\ind{j \in \nextI} \pl_j(\nextL)}^2  }}\notag\\
	&\leq \exp\inparen{ \frac{-u^2 \cdot a_t^2 }{ 2\inparen{ \frac{e\Clem}{\eta} }^2 \inparen{ \frac{  Q\log(L) 2^{t}}{L}}  }}\notag\\
	&=: \exp\inparen{ \frac{-u^2 \cdot a_t^2 }{ \delta_t^2 } }.\label{eq:probbound}
\end{align}
In the above, we useed the fact that for a Gaussian variable $g$ with variance $\sigma$, $\PR{|g|>u} \le \exp(-u^2/(2\sigma^2))$.
Now we union bound over all possible ``chains" (that is, sequences $\inset{ \pi_t(\startI,\startL) }_t$) to bound the probability that there exists a $(\startI,\startL) \in \mathcal{S}_0$ 
so that the first term $S(\startI,\startL)$ in \eqref{eq:chain} is large.  Consider the event that for all $(\startI, \startL) \in \mathcal{S}_0$,
\[  | X(\pi_t(\startI,\startL)) - X(\pi_{t+1}(\startI,\startL)) | \leq u\cdot a_t, \]
for $a_t$ to be determined shortly.  In the favorable case that this event occurs, the first term in \eqref{eq:chain} is bounded by
\[  S(\startI,\startL) = \sum_{t=0}^{t_{\max}-1} \left| X( \pi_{t}(\startI,\startL) ) - X(\pi_{t+1}(\startI,\startL) ) \right| \leq u \cdot \sum_{t=0}^{\tm - 1} a_t,\]
for all $(\startI,\startL)$.  
Let 
\begin{equation}\label{eq:defNt}
 N_t = \begin{cases} \Csize {N \choose eL/2^t}{N \choose eL/2^{t-1}} & t \geq 1 \\ {N \choose L} &t = 0 \end{cases}
\end{equation}
be our bound on $|\mathcal{S}_t|$, given by \eqref{eq:size} in Lemma \ref{lem:chaining}.
Then probability that the above good event fails to occur is at most, by the union bound,
\[ \PR{ \max_{(\startI,\startL) \in \mathcal{S}_0} S(\startI,\startL) > u \cdot \sum_{t=0}^{\tm - 1} a_t } \leq  \sum_{t=0}^{\tm-1} N_t N_{t+1} \exp\inparen{ \frac{-u^2 \cdot a_t^2}{\delta_t^2}}.\]
Indeed, there are at most $N_tN_{t+1}$ possible ``steps" between $\pi_t(\startI,\startL)$ and $\pi_{t+1}(\startI,\startL)$, and the probability that any step at level $t$ fails is given by \eqref{eq:probbound}.

Choose
\begin{equation}\label{eq:defat}
 a_t =  \sqrt{2\ln\inparen{N_tN_{t+1}}}\, \delta_t. 
\end{equation}
This choice will imply that
\begin{equation}\label{eq:expgauss}
 \EE \max_{(\startI,\startL) \in \mathcal{S}_0} S(\startI,\startL) \leq 2 \sum_{t=1}^{\tm - 1} a_t.
\end{equation}
For the reader's convenience, a brief (standard) proof of \eqref{eq:expgauss} is included in Appendix \ref{app:justexp}.
Plugging in our definition \eqref{eq:defat} of $a_t$ and then of $\delta_t$ and $N_t$ (Equations \eqref{eq:probbound} and \eqref{eq:defNt}, respectively),
\begin{align}
	 \EE \max_{(z, \Lambda) \in \mathcal{S}_0} S(\startI,\startL) &\leq 
		2\sum_{t=0}^{\tm-1}  \sqrt{ 2\ln\inparen{N_tN_{t+1}} }\,\delta_t\notag\\
		&\lesssim \sum_{t=0}^{\tm-1} \sqrt{ \frac{L}{2^t} \log(N ) } \inparen{ \frac{1}{\eta}\sqrt{ \frac{Q \log(L) 2^t}{L}}  } \notag \\
		&= \tm \inparen{\frac{\sqrt{ Q\log(N)\log(L) }}{\eta}} \notag \\
		&\leq \log^2(L) \sqrt{ Q\log(N)\log(L) },	
	 \label{eq:favorable}
\end{align}
after using the choice of $\eta = 1/\log(L)$ and $\tm \leq \log(L)$ in the final line.

With the first term $S(\startI,\startL)$ of \eqref{eq:chain} under control by \eqref{eq:favorable}, we turn to the second term, and we now
bound the probability that the final term $ X(\pi_{t_{\max}} (z,\Lambda))$ is large.
\newcommand{\maxI}{I_{\max}}
\newcommand{\maxL}{\Lambda_{\max}}
Let $(\maxI, \maxL) = \pi_{\tm}(\startI,\startL)$, so we wish to bound the Gaussian random variable
\[ X(\pi_{\tm}(\startI,\startL)) = \sum_{j \in \maxI} g_j  \pl_j (\maxL). \]
As with the increments in $S(\startI,\startL)$, we will first bound the variance of $X(\pi_{\tm}(\startI,\startL))$.
By \eqref{eq:induct}, we know that
\[ \sum_{j \in \maxI} \pl_j(\maxL) \leq Q_{\tm} \leq eQ.\]
Further, since $\pl_j(\maxL)$ is a fraction, we always have
\[ \pl_j(\maxL) \leq 1. \]
By H\"{o}lder's inequality,
\[ \sum_{j \in \maxI} \pl_j( \maxL )^2 \leq \inparen{ \sum_{j \in \maxI} \pl_j(\maxL) } \inparen{ \max_{j \in \maxI} \pl_j(\maxL) }
\leq eQ. \] 
Thus, for each $(\startI,\startL) \in \mathcal{S}_0$, $X(\pi_{\tm}(\startI,\startL))$ is a Gaussian random variable with variance at most $eQ$ (using Fact \ref{fact:2stab} in Appendix \ref{app:gauss}). 
We recall the choice from \eqref{def:tm} of
\begin{equation}\label{eq:boundtm}
 \tm = \frac{\log(L) - 2\log(1/\eta) - 2}{1 + \log(1/(1 - \eta))} \geq \log(L) - 2\log\log(L) - \Ctm,
\end{equation}
for some constant $\Ctm$, for sufficiently large $L$.
Because there are $|\mathcal{S}_{\tm}| \leq {N \choose eL/2^{\tm}}$ of these, a standard estimate for the maximum of Gaussians (see Proposition \ref{prop:maxofgauss} in Appendix \ref{app:gauss}) gives
\begin{align*}
 \EE  \max_{(\startI,\startL) \in \mathcal{S}_0} |X(\pi_{\tm}(\startI,\startL))| & \lesssim \sqrt{ \ln|\mathcal{S}_{\tm}|} \cdot \sqrt{Q} \\
&\lesssim  \sqrt{ \frac{ LQ \log(N) }{2^{\tm} } } \\
&\lesssim \log(L) \sqrt{ Q \log(N) },
\end{align*}
using the choice of $\tm$ (and the bound on it in \eqref{eq:boundtm}) in the final line.
Finally, putting together the two parts of \eqref{eq:chain}, we have
\begin{equation}\label{eq:gauss}
\EE \max_{(\startI,\startL) \in \mathcal{S}_0} X(\startI,\startL) 
\lesssim
\log^2(L) \sqrt{ Q\log(N) \log(L) } + \log(L) \sqrt{Q \log(N)}
\lesssim \log^2(L) \sqrt{Q\log(N) \log(L) }.
\end{equation}
This completes the proof of Theorem \ref{thm:gaussian} (assuming Lemma \ref{lem:chaining}).

\subsection{Proof of Lemma \ref{lem:chaining}: the desired nets exist}\label{sec:chaininglemma}
Finally, we prove Lemma \ref{lem:chaining}.
We proceed inductively.  In addition to the conclusions of the lemma, we will maintain the inductive hypotheses
\begin{equation}\label{eq:contained}
\nextI \subseteq \thisI \qquad \text{and} \qquad \nextL \subseteq \thisL
\end{equation}
for all $t$.

For the base case, $t = 0$, we set $\pi_0(\startI, \startL) = (\startI, \startL)$.
The definition of $Q$ guarantees \eqref{eq:induct}, and the definition of $\mathcal{S}_0$ guarantees \eqref{eq:Lsize}.  By definition $|\mathcal{S}_0| \leq {N \choose L}$.
Further, since by definition $I_0 = [n]$, the first part of \eqref{eq:contained} is automatically satisfied.  (We are not yet in a position to verify the base case for the second part of \eqref{eq:contained}, having not yet defined $\Lambda_1$, but we will do so shortly).

We will need to keep track of how the pluralities $\pl_j(\thisL)$ change, and for this we need the following notation.
\begin{notation}
For $\alpha \in \F_q$ and $\Lambda \subset \mathcal{C}$, let
\[ v_j(\alpha, \Lambda) = \frac{\inabs{\inset{ c \in \Lambda \suchthat c_j = \alpha }}}{|\Lambda|}\]
be the fraction of times the symbol $\alpha$ appears in the $j$'th symbol in $\Lambda$.
\end{notation}

Now we define $\mathcal{S}_t$ for $t \geq 1$.
Suppose we are given $(\thisI,\thisL) = \pi_t(\startI,\startL) \in \mathcal{S}_t$ satisfying the hypotheses of the lemma.  We need to produce $(\nextI,\nextL) \in \mathcal{S}_{t+1}$, and we will use the probabilistic method.
We will choose $\nextI$ deterministically based on $\thisL$.  Then we will choose $\nextL$ randomly, based on $\thisL$, and show that with positive probability, $(\nextI, \nextL)$ obey the desired conclusions.  Then we will fix a favorable draw of $(\nextI, \nextL)$ and call it $\pi_{t+1}(\startI,\startL)$.

We choose $\nextI$ to be the ``heavy" coordinates,
\[ \nextI := \inset{ j \suchthat |\thisL| \pl_j( \thisL) \geq \gamma}, \]
for
\begin{equation}\label{eq:gamma}
 \gamma := \frac{ 4c_1 \log(L) }{(1 - \eta)^2\eta^2}, 
\end{equation}
where $c_1$ is a suitably large constant to be fixed later. 
Notice that $\nextI$ depends only on $\thisL$ (and on $\cC$, which for the moment is fixed).

Now consider drawing $\nextL \subset \thisL$ at random by including each element of $\thisL$ in $\nextL$ independently with probability $1/2$. 
We will choose some $\nextL$ from the support of this distribution.

Before we fix $\nextL$, observe that we are already in a position to establish \eqref{eq:contained}. 
Indeed, the  second part of \eqref{eq:contained} holds for all $t$, because $\nextL \subseteq \thisL$ by construction.
To establish the first part of \eqref{eq:contained} for $t, t+1$, we use that $\thisL \subseteq \prevL$ (by induction, using \eqref{eq:contained} for $t-1,t$),  and this implies that
for all $j \in \nextI$, 
\begin{align*}
\gamma &\leq |\thisL| \pl_j(\thisL) \\
&= \max_\alpha \inabs{ \inset{ c \in \thisL \suchthat c_j = \alpha } } \\
&\leq \max_\alpha \inabs{\inset{ c \in \prevL \suchthat c_j = \alpha } } \\
&=  |\prevL| \pl_j(\prevL), 
\end{align*}
and hence $j \in \thisI$.  Thus,
\begin{equation}\label{eq:thiscontained}
\nextI \subseteq \thisI.
\end{equation}

Before we move on to the other inductive hypotheses, stated in Lemma \ref{lem:chaining},
we must fix a ``favorable" draw of $\nextL$.
In expectation, $\nextL$ behaves like $\thisL$, and so the hope is that the ``step"
\[ \pl_{\thisI}(\thisL) - \pl_{\nextI}(\nextL) \]
is small.  We quantify this in the following lemma.

\begin{lemma}\label{lem:concentration}
For all $j$, 
\[ \EE\left[ |\nextL| |\pl_j(\thisL) - \pl_j(\nextL)| \right]\leq \sqrt{\Cprop |\thisL| \log(L) \pl_j(\thisL)} \]
and
\[ \EE\left[ |\nextL|^2( \pl_j(\thisL) - \pl_j(\nextL) )^2\right] \leq \Cprop |\thisL| \log(L) \pl_j(\thisL) \]
for some constant $\Cprop$.
\end{lemma}
\begin{proof}
The second statement implies the first, by Jensen's inequality, so we prove only the second statement.  
For each $\alpha \in \F_q$, and each $j \in [n]$, consider the random variable
\begin{align*}
Y_j(\alpha) &:= |\nextL|\inparen{ v_j(\alpha, \nextL) - v_j(\alpha,\thisL)} \\
&= \sum_{c \in \thisL : c_j = \alpha} \inparen{ \xi_c - \frac{ |\nextL|}{|\thisL|} } \\
&= \sum_{c \in \thisL : c_j = \alpha} \inparen{\xi_c - \hf} + \sum_{c \in \thisL : c_j = \alpha} \inparen{ \hf  - \frac{ |\nextL|}{|\thisL|}} \\
&= \sum_{c \in \thisL : c_j = \alpha} \inparen{\xi_c - \hf} + v_j(\alpha, \thisL)  \sum_{c \in \thisL} \inparen{ \hf - \xi_c } \\
&=: Z_j(\alpha) + W_j(\alpha),
\end{align*}
where above $\xi_c$ is $1$ if $c \in \nextL$ and $0$ otherwise.
Both $Z_j(\alpha)$ and $W_j(\alpha)$ are sums of independent mean-zero random variables, and we use Chernoff bounds to control them.
First, $Z_j(\alpha)$ is a sum of $|\thisL| v_j(\alpha,\thisL)$ independent mean-zero random variables, and a Chernoff bound (Theorem \ref{thm:chernoff}) yields
\[ \PR{ |Z_j(\alpha)| > u } \leq 2\exp\inparen{ \frac{-2 u^2}{|\thisL| v_j(\alpha,\thisL)  }}
\leq 2 \exp \inparen{ \frac{ -2u^2}{|\thisL| \pl_j(\thisL)} }
.\]
Similarly, $W_j(\alpha)$ is a sum of $|\thisL|$ independent mean-zero random variables, each contained in 
\[ \left[-\frac{v_j(\alpha, \thisL)}{2} , \frac{v_j(\alpha,\thisL)}{2} \right] \subseteq \left[ -\frac{ \pl_j(\thisL)}{2}, \frac{ \pl_j(\thisL)}{2} \right],\] and we have
\[ \PR{ |W_j(\alpha)| > u } \leq 2\exp \inparen{ \frac{ -2u^2 }{ |\thisL| \pl_j(\thisL)^2 } } \leq 2\exp \inparen{ \frac{-2u^2 }{|\thisL|\pl_j(\thisL)}},\]
using the fact that $\pl_j(\thisL) \leq 1$. 
Together,
\[ \PR{ |Y_j(\alpha)| > u} \leq \PR{ |W_j(\alpha)| > u/2 } + \PR{ |Z_j(\alpha)| > u/2 } \leq 4\exp\inparen{ \frac{ -u^2 }{ 2\pl_j(\thisL)|\thisL| } }, \]
Let 
\[ T_j = \inset{ \alpha \in \F_q \suchthat \exists c \in \thisL, c_j = \alpha } \]
be the set of symbols that show up in the $j$'th coordinates of $\thisL$.  Then
\[ |T_j| \leq \min\{q, |\thisL|\} \leq L. \]
By the union bound, and letting $v = u^2$,
\begin{equation}\label{eq:orig_chernoff}
\PR{ \max_{ \alpha \in \F_q } Y_j(\alpha)^2 > v } = \PR{ \max_{ \alpha \in T_j } Y_j(\alpha)^2 > v } \leq 4L \exp\inparen{ \frac{-v }{2 \pl_j(\thisL) |\thisL| }}.
\end{equation}
Next, we show that if all of the $Y_j(\alpha)$ are under control, then so are the pluralities $\pl_j(\thisL)$.
For any four numbers $A,B,C,D$ with $A \leq B$ and $C \leq D$, we have
\begin{equation}\label{eq:erg}
 |B - D| \leq \max\inset{ |B - C|, |D - A| }. 
\end{equation}
Indeed, we have
\[ B - D \leq (B - D) + (D - C) = B - C \qquad \text{and} \qquad  D - B \leq (D - B) + (B - A) = D - A. \]
The claim \eqref{eq:erg} follows.
Now, for fixed $j$, let 
\[\alpha =  \argmax_{\sigma \in T_j} v_j( \sigma, \thisL) \qquad \text{and} \qquad \beta = \argmax_{ \sigma \in T_j} v_j(\sigma,\nextL),\]
so that
\[ |\nextL| v_j(\alpha, \nextL) \leq |\nextL| v_j(\beta,\nextL) \qquad \text{and} \qquad |\nextL| v_j(\beta, \thisL) \leq |\nextL| v_j(\alpha, \thisL).\]
By \eqref{eq:erg}, we have
\begin{align*}
|\nextL| | \pl_j(\nextL) - \pl_j(\thisL) | &= |\nextL| | v_j(\beta,\nextL) - v_j(\alpha,\thisL) | \\
&\leq |\nextL| \max\inset{ |v_j(\alpha, \thisL) - v_j(\alpha, \nextL)|, |v_j(\beta, \thisL) - v_j(\beta,\nextL) |}\\
&\leq \max_{\alpha \in T_j} |Y_j(\alpha)|.
\end{align*}
Thus, the probability that $|\pl_j(\nextL) - \pl_j(\thisL)|$ is large is no more than the probability that $\max_{\alpha \in T_j} |Y_j(\alpha)|$ is large, 
and we conclude from \eqref{eq:orig_chernoff} that
\[ \PR{ |\nextL|^2 (\pl_j(\thisL) - \pl_j(\nextL))^2 > v } \leq 4L \exp\inparen{ \frac{-v }{2 \pl_j(\thisL) |\thisL|} }.\]
Integrating, we bound the expectation by
\begin{align*}
	\EE  |\nextL|^2 (\pl_j(\thisL) - \pl_j(\nextL))^2 &= \int_0^\infty \PR{ \max_{\alpha \in T_j} Y_j(\alpha)^2 > v}\,dv\\
	&\leq A + 4L \int_A^\infty \exp \inparen{ \frac{-v}{2 \pl_j(\thisL)|\thisL| }} \,dv\\
	&= A + 4L \cdot 2\pl_j(\thisL) |\thisL| \cdot \exp\inparen{ \frac{-A }{2 \pl_j(\thisL) |\thisL|}}
\end{align*}
for any $A \geq 0$.  Choosing $A = 2\pl_j(\thisL)|\thisL| \ln(4L)$ gives
\[ \EE  |\nextL|^2 (\pl_j(\thisL) - \pl_j(\nextL))^2 \leq 2|\thisL|\pl_j(\thisL)\inparen{\ln(4L) + 1}.\]
Setting $\Cprop$ correctly proves the second item in Lemma \ref{lem:concentration}, and the first follows from Jensen's inequality.
\end{proof}

The next lemma uses Lemma \ref{lem:concentration} to argue that a number of good things happen all at once.
\begin{lemma}\label{lem:almostdone}
There is some $\nextL \subseteq \thisL$ so that:
\begin{enumerate}
	\item \label{item:sizeL}
\[ \inparen{ \frac{1 - \eta}{2} }^{t+1} L \leq \inparen{ \frac{1 - \eta}{2} }|\thisL| \leq |\nextL| \leq \inparen{ \frac{1 + \eta}{2} }|\thisL| \leq \inparen{ \frac{1 + \eta}{2}}^{t+1} L. \]
	\item \label{item:Q}
	\[ \sum_{j \in \nextI} \pl_j(\nextL) \leq \sum_{j \in \nextI} \pl_j(\thisL) + \sum_{j\in \nextI} \sqrt{  \frac{c_1 |\thisL| \log(L) \pl_j(\thisL)}{|\nextL|^2}}  \]
	\item \label{item:mainterm} 
	\[ \inparen{ \sum_{j \in \nextI} (\pl_j(\nextL) - \pl_j(\thisL))^2 }^{1/2} \leq \frac{\sqrt{c_1 |\thisL| \log(L) Q_t}}{|\nextL|}\]
\end{enumerate}
for some constant $c_1$.
\end{lemma}

\begin{proof}
We show that (for an appropriate choice of $c_1$), 
each of these items occurs with probability at least $2/3$, $3/4$, and $3/4$, respectively.
Thus, all three occur with probability at least $1/6$, and in particular there is a set $\nextL$ which satisfies all three.

First, we address Item \ref{item:sizeL}.  By a Chernoff bound,
\[ \PR{ \inabs{|\nextL| - \frac{1}{2}|\thisL|} > u } \leq 2\exp \inparen{ - 2u^2 /|\thisL| }, \]
By the inductive hypothesis \eqref{eq:Lsize}, 
\[ |\thisL| \geq \inparen{ \frac{ 1 - \eta}{2}}^t L, \]
and so by our choice of $\tm$ and the fact that $t \leq \tm$, we have
\begin{equation}\label{eq:Lbd}
 |\thisL| \geq 4/\eta^2. 
\end{equation}
Thus, 
\[ \PR{ \inabs{ |\nextL| - \frac{|\thisL|}{2} } \geq \frac{ \eta |\thisL|}{2}  } \leq 2e^{-2} < 1/3. \]
Again by the inductive hypothesis \eqref{eq:Lsize} applied to $|\thisL|$, 
we conclude that
\[ \inparen{ \frac{1 - \eta}{2} }^{t+1} L \leq \inparen{ \frac{1 - \eta}{2} }|\thisL| \leq |\nextL| \leq \inparen{ \frac{1 + \eta}{2} }|\thisL| \leq \inparen{ \frac{1 + \eta}{2}}^{t+1} L. \]

For Item \ref{item:Q}, we invoke
Lemma \ref{lem:concentration} and linearity of expectation to obtain
\[ \EE \sum_{j \in \nextI} |\nextL| | \pl_j(\thisL) - \pl_j(\nextL)| \leq \sum_{j \in \nextI} \sqrt{ \Cprop\log(L)\pl_j(\thisL)|\thisL|}.\]
By Markov's inequality, as long as $c_1 \geq 16\Cprop$, with probability at least $3/4$, 
\[ \sum_{j \in \nextI} |\nextL| | \pl_j(\thisL) - \pl_j(\nextL)| \leq \sum_{j \in \nextI} \sqrt{ c_1 \log(L)\pl_j(\thisL)|\thisL|},\]
and in the favorable case the triangle inequality implies
\begin{align*}
\sum_{j \in \nextI} \pl_j(\nextL) &\leq \sum_{j \in \nextI}  \pl_j(\thisL) + \sum_{j \in \nextI} | \pl_j(\thisL) - \pl_j(\nextL)| \\
&\leq \sum_{j \in \nextI} \pl_j(\thisL) + \sum_{j \in \nextI} \frac{ \sqrt{c_1 \log(L) \pl_j(\thisL)|\thisL|} }{|\nextL|}.
\end{align*}
Thus, Item \ref{item:Q} holds with probability at least $3/4$.

Similarly, for Item \ref{item:mainterm}, Lemma \ref{lem:concentration} and linearity of expectation (as well as Jensen's inequality) implies that 
\begin{align*}
\EE \inparen{ \sum_{j \in \nextI} |\nextL|^2 (\pl_j(\nextL) - \pl_j(\thisL) )^2 }^{1/2}  &\leq 
\inparen{ \sum_{j \in \nextI}   \Cprop |\thisL| \log(L) \pl_j(\thisL) }^{1/2}\\
&\leq \inparen{ \sum_{j \in \thisI}   \Cprop |\thisL| \log(L) \pl_j(\thisL) }^{1/2} \qquad \text{ since $\nextI \subseteq \thisI$}\\
&\leq \sqrt{ \Cprop |\thisL| \log(L) Q_t } \qquad \text{ by the inductive hypothesis \eqref{eq:induct} }.
\end{align*}
Again, Markov's inequality and an appropriate restriction on $c_1$ implies that Item \ref{item:mainterm} occurs with probability strictly more than $3/4$.  

This concludes the proof of Lemma \ref{lem:almostdone}.
\end{proof}

Finally, we show how Lemma \ref{lem:almostdone} implies the conclusions of Lemma \ref{lem:chaining} for $t+1$, notably 
\eqref{eq:induct},
\eqref{eq:Lsize}
and \eqref{eq:width}.
First, we observe that \eqref{eq:Lsize} follows immediately from Lemma \ref{lem:almostdone}, Item \ref{item:sizeL}.
Next we consider \eqref{eq:induct}.  The definition of $\nextI$ and the choice of $\gamma$, along with the fact from Lemma \ref{lem:almostdone}, Item \ref{item:sizeL} that
$|\nextL| \geq \inparen{ \frac{1 - \eta}{2} }|\thisL|$,
 imply that for $j \in \nextI$, 
\[ |\thisL| \pl_j(\thisL) \geq \gamma \geq \inparen{ \frac{ |\thisL|}{|\nextL|}}^2 \frac{ c_1 \log(L) }{\eta^2}, \]
and so 
\begin{equation}\label{eq:nextIgood}
 \frac{ \sqrt{ c_1 |\thisL| \log(L) \pl_j(\thisL)} }{|\nextL|} \leq  \eta \pl_j(\thisL). 
\end{equation}
Thus, 
\begin{align*}
\sum_{j \in \nextI} \pl_j(\nextL) &\leq \sum_{j \in \nextI} \inparen{ 1 + \eta} \pl_j(\thisL) \qquad \text{ by Lemma \ref{lem:almostdone}, Item \ref{item:Q} and from \eqref{eq:nextIgood}} \\
&\leq \inparen{ 1 + {\eta} } \sum_{j \in \thisI} \pl_j(\thisL) \qquad \text{ since $\nextI \subseteq \thisI$, by \eqref{eq:thiscontained} }\\
&\leq \inparen{ 1 + { \eta}} Q_t \qquad \text{ by the inductive hypothesis \eqref{eq:induct} for $t$ } \\
& = \inparen{ 1 + \eta }^{t+1} Q \qquad \text{ by the definition of $Q_t$ }\\
&= Q_{t+1}.
\end{align*}
This establishes \eqref{eq:induct}.

To establish the distance criterion \eqref{eq:width}, we use the triangle inequality to write
\begin{align}
\| \pl_{\thisI}(\thisL) - \pl_{\nextI}(\nextL) \|_2 &= \| \pl_{\nextI}(\thisL) + \pl_{\thisI \setminus \nextI}(\thisL) - \pl_{\nextI}(\nextL) \|_2  \label{eq:monster}\\
&\leq  \| \pl_{\nextI}(\thisL) - \pl_{\nextI}(\nextL) \|_2 \label{eq:goodguys}\\
&\qquad\qquad + \| \pl_{\thisI \setminus \nextI} (\thisL) \|_2 \label{eq:stepping}
\end{align}
The first term \eqref{eq:goodguys} is bounded by Lemma \ref{lem:almostdone}, Item \ref{item:mainterm}, by
\[ \| \pl_{\nextI}(\thisL) - \pl_{\nextI}(\nextL) \|_2 \leq  \frac{ \sqrt{ c_1 |\thisL| \log(L) Q_t }}{|\nextL|}. \]
To bound \eqref{eq:stepping}, we will bound both the $\ell_\infty$ and $\ell_1$ norms of $\pl_{\thisI \setminus \nextI} (\thisL)$ and use H\"{o}lder's inequality to control the $\ell_2$ norm.  
By the inductive hypothesis \eqref{eq:induct} and the fact  \eqref{eq:thiscontained} that $\nextI \subseteq \thisI$, 
\[ \|\pl_{\thisI \setminus \nextI}(\thisL) \|_1 \leq \|\pl_{\thisI}(\thisL)\|_1 \leq Q_t. \]
Also, by the definition of $\nextI$, 
\[ \|\pl_{\thisI \setminus \nextI}(\thisL) \|_\infty \leq \frac{\gamma}{|\thisL|}.\]
Together, H\"{o}lder's inequality implies that
\[ \|\pl_{\thisI \setminus \nextI}(\thisL) \|_2 \leq \sqrt{ \|\pl_{\thisI \setminus \nextI}(\thisL)\|_1 \|\pl_{\thisI \setminus \nextI}(\thisL)\|_\infty } \leq \sqrt{ \frac{ \gamma Q_t }{|\thisL|} }.\]
This bounds the second term \eqref{eq:stepping} of \eqref{eq:monster}, and putting it all together we have
\begin{align*}
\| \pl_{\thisI}(\thisL) - \pl_{\nextI}(\nextL) \|_2 &\leq   \frac{ \sqrt{ c_1 |\thisL| \log(L) Q_t }}{|\nextL|} +  \sqrt{ \frac{ \gamma Q_t }{|\thisL|} }.
\end{align*}
Using the fact from Lemma \ref{lem:almostdone}, Item \ref{item:sizeL} that $|\thisL|/|\nextL| \leq 2/(1 - \eta)$, as well as the definition of $\gamma$ in \eqref{eq:gamma}, we may bound the above expression by
\begin{align*}
\| \pl_{\thisI}(\thisL) - \pl_{\nextI}(\nextL) \|_2 &\leq  \inparen{1 +\frac{1}{\eta}} \inparen{ \frac{2}{1 - \eta}} \sqrt{\frac{ c_1 \log(L) Q_t }{|\thisL|}}.
\end{align*}
This establishes \eqref{eq:width}, for an appropriate choice of $\Clem$, and for sufficiently large $L$ (and hence sufficiently small $\eta$).

Finally, we verify the condition \eqref{eq:size} on the size $|\mathcal{S}_{t+1}|$.
By \eqref{eq:Lsize}, and the fact that our choices of $\eta$ and $\tm$ imply that $(1 + \eta)^t \leq e$, 
$|\thisL| \leq eL/2^t. $
We saw earlier that $\nextI$ depends only on $\thisL$, so (using the fact that $L \leq N/2$), there are at most 
\[\sum_{r = 1}^{eL/2^t} {N \choose r}  \lesssim {N \choose {eL/2^t}} \]
choices for $\nextI$.
Similarly, we just chose $\nextL$ so that 
$|\nextL| \leq  eL/2^{t+1},$ so
there are at most
$\sum_{r=1}^{eL/2^t}  {N \choose r} \lesssim { N \choose {eL/2^{t+1}} } $
choices for $\nextL$.  Altogether, there are at most
\[ \Csize {N \choose eL/2^t }{N \choose eL/2^{t+1} } \]
choices for the pair $(\nextI, \nextL)$, for an appropriate constant $\Csize$, and this establishes \eqref{eq:Lsize}.

This completes the proof of Lemma \ref{lem:chaining}.

\section{Conclusion and future work}
\label{sec:concl}

We have shown that ``most" Reed-Solomon codes are list decodable beyond the Johnson bound, answering a long-standing open question (Question~\ref{q:rs}) of~\cite{GS98,venkat-thesis,atri-thesis,salil-book}.  More precisely, we have shown that with high probability, a Reed-Solomon code with random evaluation points of rate 
\[ \Omega\inparen{ \frac{\eps}{\log(q) \log^5(1/\eps) } } \]
is list decodable up to a $1 - \eps$ fraction of errors with list size $O(1/\eps)$.  This beats the Johnson bound whenever $\eps \leq \tilde{O}\inparen{ 1/\log(q)}$.

Our proof actually applies more generally to randomly punctured codes, 
and provides a positive answer to our second motivating question, Question~\ref{q:largeq}, about 
whether randomly punctured codes with good distance can beat the Johnson bound.
As an added corollary, we have obtained improved bounds on the list decodability of random linear codes over large alphabets.   Our bounds are nearly optimal (up to polylogarithmic factors), and are the best known whenever $q \gtrsim \log^5(1/\eps)$.

The most obvious open question that remains is to remove the polylogarithmic
factors from the rate bound.  The factor of $\log(q)$ is especially
troublesome: it bites when $q = 2^{\Omega(1/\eps)}$ is very large, but this
parameter regime can be reasonable for Reed-Solomon codes.  Removing this logarithmic
factor seems as though it may require a restructuring of the argument.  A
second question is to resolve the discrepancy between our upper bound on
list sizes and the bound associated with general random codes of the same
rate; there is a gap of a factor of $\eps$ in the parameter regime $1/\eps \leq q \leq 1/\eps^2$.

To avoid ending on the shortcomings of our argument, we mention a few hopeful
directions for future work.  Our argument applies to randomly punctured codes
in general, and it is natural to ask for more examples of codes where Theorem
\ref{thm:mainthm} can improve the status quo.  Additionally, list decodable
codes are connected to many other pseudorandom objects; it would be extremely
interesting to explore the ramifications of our argument for random families of
extractors or expanders, for example.

\section*{Acknowledgments}
We are very grateful to Mahmoud Abo Khamis, Venkat Guruswami, Yi Li, Anindya Patthak, and Martin Strauss for careful proof-reading and helpful comments. We thank Venkat for suggesting the title of the paper. AR would like to thank Swastik Kopparty and Shubhangi Saraf for some discussions on related questions at \href{http://www.dagstuhl.de/en/program/calendar/semhp/?semnr=12421}{Dagstuhl} that led to questions considered in this paper.

\bibliographystyle{alpha}
\bibliography{refs}

\begin{thebibliography}{BSKR10}

\bibitem[BSKR10]{BKR10}
Eli Ben-Sasson, Swastik Kopparty, and Jaikumar Radhakrishnan.
\newblock Subspace polynomials and limits to list decoding of reed-solomon
  codes.
\newblock {\em IEEE Transactions on Information Theory}, 56(1):113--120, 2010.

\bibitem[CGV13]{cgv2012}
Mahdi Cheraghchi, Venkatesan Guruswami, and Ameya Velingker.
\newblock Restricted isometry of fourier matrices and list decodability of
  random linear codes.
\newblock In {\em Proceedings of the Twenty-Fourth Annual ACM-SIAM Symposium on
  Discrete Algorithms (SODA)}, pages 432--442, 2013.

\bibitem[CPS99]{CPS99}
{Jin-yi} Cai, Aduri Pavan, and D.~Sivakumar.
\newblock On the hardness of permanent.
\newblock In {\em Proceedings of the 16th Annual Symposium on Theoretical
  Aspects of Computer Science (STACS)}, pages 90--99, 1999.

\bibitem[CW07]{CW07}
Qi~Cheng and Daqing Wan.
\newblock On the list and bounded distance decodability of reed-solomon codes.
\newblock {\em SIAM J. Comput.}, 37(1):195--209, 2007.

\bibitem[DL12]{DL12}
Zeev Dvir and Shachar Lovett.
\newblock Subspace evasive sets.
\newblock In {\em Proceedings of the 44th Symposium on Theory of Computing
  Conference (STOC)}, pages 351--358, 2012.

\bibitem[Eli57]{elias}
Peter Elias.
\newblock List decoding for noisy channels.
\newblock {\em Technical Report 335, Research Laboratory of Electronics, MIT},
  1957.

\bibitem[GHK11]{GHK11}
Venkatesan Guruswami, Johan H{\aa}stad, and Swastik Kopparty.
\newblock On the list-decodability of random linear codes.
\newblock {\em IEEE Transactions on Information Theory}, 57(2):718--725, 2011.

\bibitem[GK13]{GK13}
Venkatesan Guruswami and Swastik Kopparty.
\newblock Explicit subspace designs.
\newblock In {\em FOCS}, 2013.
\newblock To appear.

\bibitem[GKZ08]{GKZ08}
Parikshit Gopalan, Adam~R. Klivans, and David Zuckerman.
\newblock List-decoding reed-muller codes over small fields.
\newblock In {\em Proceedings of the 40th Annual ACM Symposium on Theory of
  Computing (STOC)}, pages 265--274, 2008.

\bibitem[GN13]{gurnar2013}
Venkatesan Guruswami and Srivatsan Narayanan.
\newblock Combinatorial limitations of average-radius list decoding.
\newblock {\em RANDOM}, 2013.

\bibitem[Gop10]{G10}
Parikshit Gopalan.
\newblock A fourier-analytic approach to reed-muller decoding.
\newblock In {\em Proceedings of the 51th Annual IEEE Symposium on Foundations
  of Computer Science (FOCS)}, pages 685--694, 2010.

\bibitem[GR06]{GR06}
Venkatesan Guruswami and Atri Rudra.
\newblock Limits to list decoding reed-solomon codes.
\newblock {\em IEEE Transactions on Information Theory}, 52(8):3642--3649,
  2006.

\bibitem[GR08]{GR08}
Venkatesan Guruswami and Atri Rudra.
\newblock Explicit codes achieving list decoding capacity: Error-correction
  with optimal redundancy.
\newblock {\em IEEE Transactions on Information Theory}, 54(1):135--150, 2008.

\bibitem[GS98]{GS98}
Venkatesan Guruswami and Madhu Sudan.
\newblock Improved decoding of reed-solomon and algebraic-geometric codes.
\newblock In {\em Proceedings of 39th Annual Symposium on Foundations of
  Computer Science (FOCS)}, pages 28--39, 1998.

\bibitem[GS99]{GS99}
Venkatesan Guruswami and Madhu Sudan.
\newblock Improved decoding of reed-solomon and algebraic-geometry codes.
\newblock {\em IEEE Transactions on Information Theory}, 45(6):1757--1767,
  1999.

\bibitem[GS01]{gs2001}
Venkatesan Guruswami and Madhu Sudan.
\newblock Extensions to the johnson bound, 2001.

\bibitem[GS03]{GS03}
Venkatesan Guruswami and Igor Shparlinski.
\newblock Unconditional proof of tightness of johnson bound.
\newblock In {\em Proceedings of the Fourteenth Annual ACM-SIAM Symposium on
  Discrete Algorithms (SODA)}, pages 754--755, 2003.

\bibitem[Gur04]{venkat-thesis}
Venkatesan Guruswami.
\newblock {\em List Decoding of Error-Correcting Codes (Winning Thesis of the
  2002 ACM Doctoral Dissertation Competition)}, volume 3282 of {\em Lecture
  Notes in Computer Science}.
\newblock Springer, 2004.

\bibitem[GV10]{GV10}
Venkatesan Guruswami and Salil Vadhan.
\newblock A lower bound on list size for list decoding.
\newblock {\em Information Theory, IEEE Transactions on}, 56(11):5681--5688,
  2010.

\bibitem[GW13]{GW13}
Venkatesan Guruswami and Carol Wang.
\newblock Linear-algebraic list decoding for variants of reed-solomon codes.
\newblock {\em IEEE Transactions on Information Theory}, 59(6):3257--3268,
  2013.

\bibitem[GX12]{GX12}
Venkatesan Guruswami and Chaoping Xing.
\newblock Folded codes from function field towers and improved optimal rate
  list decoding.
\newblock In {\em Proceedings of the 44th Symposium on Theory of Computing
  Conference (STOC)}, pages 339--350, 2012.

\bibitem[GX13]{GX13}
Venkatesan Guruswami and Chaoping Xing.
\newblock List decoding reed-solomon, algebraic-geometric, and gabidulin
  subcodes up to the singleton bound.
\newblock In {\em Proceedings of the 45th ACM Symposium on the Theory of
  Computing (STOC)}, pages 843--852, 2013.

\bibitem[Kop12]{K12}
Swastik Kopparty.
\newblock List-decoding multiplicity codes.
\newblock {\em Electronic Colloquium on Computational Complexity (ECCC)},
  19:44, 2012.

\bibitem[LT91]{lt}
Michel Ledoux and Michel Talagrand.
\newblock {\em Probability in Banach Spaces: isoperimetry and processes},
  volume~23.
\newblock Springer, 1991.

\bibitem[MS77]{MS77}
F.~J. MacWilliams and N.~J.~A. Sloane.
\newblock {\em The theory of error correcting codes / F.J. MacWilliams, N.J.A.
  Sloane}.
\newblock North-Holland Pub. Co. ; sole distributors for the U.S.A. and Canada,
  Elsevier/North-Holland Amsterdam ; New York : New York, 1977.

\bibitem[PV05]{PV05}
Farzad Parvaresh and Alexander Vardy.
\newblock Correcting errors beyond the guruswami-sudan radius in polynomial
  time.
\newblock In {\em Proceedings of the 46th Annual IEEE Symposium on Foundations
  of Computer Science (FOCS)}, pages 285--294, 2005.

\bibitem[Rud97]{rudelson97}
Mark Rudelson.
\newblock Contact points of convex bodies.
\newblock {\em Israel Journal of Mathematics}, 101(1):93--124, 1997.

\bibitem[Rud07]{atri-thesis}
Atri Rudra.
\newblock {\em List decoding and property testing of error-correcting codes}.
\newblock PhD thesis, University of Washington, 2007.

\bibitem[RV08]{rv08}
Mark Rudelson and Roman Vershynin.
\newblock On sparse reconstruction from fourier and gaussian measurements.
\newblock {\em Communications on Pure and Applied Mathematics},
  61(8):1025--1045, 2008.

\bibitem[Sud97]{sudan}
Madhu Sudan.
\newblock Decoding of reed solomon codes beyond the error-correction bound.
\newblock {\em J. Complexity}, 13(1):180--193, 1997.

\bibitem[Sud00]{madhu-survey}
Madhu Sudan.
\newblock List decoding: algorithms and applications.
\newblock {\em SIGACT News}, 31(1):16--27, 2000.

\bibitem[Tal05]{genchain}
Michel Talagrand.
\newblock {\em The generic chaining: upper and lower bounds for stochastic
  processes}.
\newblock Springer, 2005.

\bibitem[Vad12]{salil-book}
Salil~P. Vadhan.
\newblock Pseudorandomness.
\newblock {\em Foundations and Trends in Theoretical Computer Science},
  7(1-3):1--336, 2012.

\bibitem[Woo13]{woot2013}
Mary Wootters.
\newblock On the list decodability of random linear codes with large error
  rates.
\newblock In {\em Proceedings of the 45th annual ACM symposium on Symposium on
  theory of computing}, pages 853--860. ACM, 2013.

\bibitem[Woz58]{wozencraft}
John~M. Wozencraft.
\newblock List {D}ecoding.
\newblock {\em Quarterly Progress Report, Research Laboratory of Electronics,
  MIT}, 48:90--95, 1958.

\end{thebibliography}

\appendix
\section{Proofs of Corollaries \ref{cor:smallq} and \ref{cor:largeq}}\label{app:avgjb}
In this appendix, we first prove a few variants on the Johnson bound, which are needed for the proofs of Corollaries \ref{cor:smallq} and \ref{cor:largeq}.
We require average-radius versions of two statements of the Johnson bound, found in~\cite{gs2001} and~\cite{MS77}, respectively.
It appears to be folklore that such statements are true (and follow from the proofs in the two works cited above).  For completeness, we
include proofs below in Section \ref{sapp:avgjb}.
Finally, we prove Corollaries \ref{cor:smallq} and \ref{cor:largeq} in Section \ref{sapp:cors}.

\subsection{Average radius Johnson bounds}\label{sapp:avgjb}
\begin{theorem}\label{thm:jb} Let $\mathcal{C}:\F_q^k \to \F_q^n$ be any code.  Then for all $\Lambda \subset \F_q^k$ of size $L$ and for all $z \in \F_q^n$, 
\[ \sum_{x \in \Lambda} \agr( c(x), z) \leq \frac{nL}{q} + \frac{ nL}{2\eps} \inparen{ 1 + \eps^2 }\inparen{1 - \frac{1}{q}} - \frac{n}{2L\eps}\sum_{x \neq y \in \Lambda} d(c(x),c(y)). \]
\end{theorem}

\begin{remark} An average-radius $q$-ary Johnson bound follows from Theorem \ref{thm:jb} by bounding $d(c(x),c(y))$ by $1 - 1/q - \eps^2$.  In this case, the theorem implies that any code of distance $1 - 1/q - \eps^2$ is average-radius list decodable up to error rate $\rho = 1 - 1/q - \eps$ as long as the list size $L$ obeys $L \geq 2/\eps^2$.
\end{remark}

\begin{proof}
Fix a $z\in\F_q^n$. 
The crux of the proof is to map the relevant vectors over $\F_q^n$ to vectors in $\R^{nq}$ as follows.  Given a vector $u\in\F_q^n$, let $u'\in\R^{nq}$ denote the concatenation 
\[ u' = ( e_{u_1}, e_{u_2}, \ldots, e_{u_n} ), \]
where $e_{u_i} \in \{0,1\}^q$ is the vector which is one in the $u_i$'th index and zero elsewhere.
(Above, we fix an arbitrary mapping of $\F_q$ to $[q]$). 
In particular, for an
$x\in \Lambda$, we will use $c'(x)$ to denote the mapping of the codeword
$c(x)$. Finally let $v\in\R^{nq}$ be
\[ v = \eps\cdot z'+\inparen{  \frac{ 1 - \eps }{q} } \cdot \ind{}, \]
where $\ind{}$ denotes the all-ones vector.

Given the definitions above, it can be verified that the identities below hold for every $x\neq y\in\Lambda$:
\begin{equation}
\label{eq:jb-c-v}
\ip{c'(x)}{v}=\eps\cdot \agr(c(x),z)+\frac{(1-\eps)n}{q},
\end{equation}

\begin{equation}
\label{eq:jb-v-v}
\ip{v}{v}=\frac{n}{q}+\eps^2\left(1-\frac{1}{q}\right)n,
\end{equation}

\begin{equation}
\label{eq:jb-c-c-diff}
\ip{c'(x)}{c'(y)} = n(1-d(c(x),c(y)),
\end{equation}

and

\begin{equation}
\label{eq:jb-c-c-same}
\ip{c'(x)}{c'(x)}=n.
\end{equation}

Now consider the following sequence of relations:
{\allowdisplaybreaks
\begin{align}
\label{step:jb-1}
0&\le \ip{\sum_{x\in\Lambda} \inparen{c'(x)-v}}{\sum_{x\in\Lambda} \inparen{c'(x)-v}}\\
&=\sum_{x,y\in\Lambda} \ip{c'(x)}{c'(y)} -\sum_{x,y\in\Lambda}(\ip{c'(x)}{v}+\ip{c'(y)}{v})+\sum_{x,y\in\Lambda}\ip{v}{v} \nonumber\\
&=\sum_{x\in\Lambda} \ip{c'(x)}{c'(x)} +\sum_{x\neq y\in\Lambda} \ip{c'(x)}{c'(y)}- 2L\cdot \sum_{x\in\Lambda}\ip{c'(x)}{v} +\sum_{x,y\in\Lambda}\ip{v}{v} \nonumber\\
\label{step:jb-2}
&=nL +n\sum_{x\neq y\in\Lambda} (1-d(c(x),c(y)))-2L\cdot\sum_{x\in\Lambda}\inparen{\eps\cdot \agr(c(x),z)+\frac{(1-\eps)n}{q}}+L^2\cdot\inparen{\frac{n}{q}+\eps^2\left(1-\frac{1}{q}\right)n}\\
&=nL^2\cdot\inparen{1+\frac{1}{q}+\eps^2\left(1-\frac{1}{q}\right)-\frac{2(1-\eps)}{q}} -n\sum_{x\neq y\in\Lambda} d(c(x),c(y))-2L\eps\cdot\sum_{x\in\Lambda} \agr(c(x),z) \nonumber\\
\label{step:jb-last}
&=nL^2\cdot\inparen{(1+\eps^2)\left(1-\frac{1}{q}\right)+\frac{2\eps}{q}} -n\sum_{x\neq y\in\Lambda} d(c(x),c(y))-2L\eps\cdot\sum_{x\in\Lambda} \agr(c(x),z)
\end{align}
}
In the above, (\ref{step:jb-1}) follows from the fact that the norm of a vector is always positive and (\ref{step:jb-2}) follows from (\ref{eq:jb-c-v}), (\ref{eq:jb-v-v}), (\ref{eq:jb-c-c-diff}) and (\ref{eq:jb-c-c-same}).

Equation (\ref{step:jb-last}) then implies that
\[2L\eps\cdot\sum_{x\in\Lambda} \agr(c(x),z) \le nL^2\cdot\inparen{(1+\eps^2)\left(1-\frac{1}{q}\right)+\frac{2\eps}{q}} -n\sum_{x\neq y\in\Lambda} d(c(x),c(y)),\]
which implies the statement after rearranging terms.
\end{proof}

Next, we prove a second average-radius variant of the Johnson bound, which has been copied almost verbatim from~\cite{MS77}.

\begin{theorem}
\label{thm:jb2} Let $\mathcal{C}:\F_q^k \to \F_q^n$ be any code.  Then for all $\Lambda \subset \F_q^k$ of size $L$ and for all $z \in \F_q^n$, 
\[\sum_{x\in\Lambda} \agr(c(x),z) \le \frac{1}{2}\inparen{n+\sqrt{n^2+ 4n^2L(L-1)-4n^2\sum_{x\neq y\in\Lambda} d(c(x),c(y))}}.\]
\end{theorem}
\begin{proof}
For every $j\in [n]$, define
\[a_j= |\inset{x\in\Lambda| c(x)_j=z_j}|.\]
Note that
\begin{equation}
\label{eq:sum-agr}
\sum_{j=1}^n a_j=\sum_{x\in\Lambda} \agr(c(x),z),
\end{equation}
and
\begin{align}
\sum_{j=1}^n \binom{a_j}{2}&= \frac{1}{2} \cdot \sum_{j=1}^n \sum_{x \neq y \in \Lambda} \ind{c(x)_j = z_j}\ind{c(y)_j = z_j} \notag\\
&\leq \sum_{j=1}^n \sum_{x \neq y \in \Lambda} \ind{c(x)_j = c(y)_j}\notag\\
&= \frac{1}{2}\cdot \sum_{x\neq y\in\Lambda} \agr(c(x),c(y)) \notag \\
\label{eq:sumsq-agr}
&=\frac{L(L-1)n}{2}-\frac{n}{2}\sum_{x\neq y\in\Lambda} d(c(x),c(y)).
\end{align}
Next, note that by the Cauchy-Schwartz inequality,
\[ \sum_{j=1}^n \binom{a_i}{2} =\frac{1}{2}\inparen{\sum_{j=1}^n a_j^2 -\sum_{j=1}^n a_j} \ge \frac{1}{2n}\inparen{\sum_{j=1}^n a_j}^2 -\frac{1}{2}\sum_{j=1}^n a_j.\]
Combining the above with (\ref{eq:sum-agr}) and (\ref{eq:sumsq-agr}) implies that
\[\inparen{\sum_{x\in\Lambda} \agr(c(x),z)}^2 -n\cdot \sum_{x\in\Lambda} \agr(c(x),z) -\inparen{ n^2L(L-1)-n^2\sum_{x\neq y\in\Lambda} d(c(x),c(y))}\le 0,\]
which in turn implies (by the fact that the sum we care about lies in between the two roots of the quadratic equation) that
\[\sum_{x\in\Lambda} \agr(c(x),z) \le \frac{1}{2}\inparen{n+\sqrt{n^2+ 4n^2L(L-1)-4n^2\sum_{x\neq y\in\Lambda} d(c(x),c(y))}},\]
which completes the proof.
\end{proof}


\subsection{Proofs of Corollaries \ref{cor:smallq} and \ref{cor:largeq}}
\label{sapp:cors}
The proofs of both Corollaries follow essentially from the proofs of the Johnson bound in Section \ref{sapp:avgjb}.  We use two versions of the Johnson bound,
one from~\cite{gs2001} which is more useful for our ``small $q$" regime, and another proof from~\cite{MS77} which produces better
results in the ``large $q$" regime.  

\begin{proof}[Proof of Corollary \ref{cor:smallq}] 
Suppose that $L \geq 2/\eps^2$ and that the distance of $\mathcal{C}'$ is at least $1 - 1/q - \eps^2/2$.
We need an average-radius version of the Johnson bound, which we provide in Theorem \ref{thm:jb} in Appendix \ref{sapp:avgjb}.
By Theorem \ref{thm:jb}, for any $z \in \F_q^n$ and for all $\Lambda \subset \F_q^k$ of size $L$, 
\begin{equation}\label{eq:jb}
 \sum_{x \in \Lambda} \agr( c(x), z) \leq \frac{nL}{q} + \frac{ nL}{2\eps} \inparen{ 1 + \eps^2 }\inparen{1 - \frac{1}{q}} - \frac{n}{2L\eps}\sum_{x \neq y \in \Lambda} d(c(x),c(y)). 
\end{equation}
By Theorem \ref{thm:mainthm}, it suffices to control $\mathcal{E}$. 
Since the right hand side above does not depend on $z$,
{\allowdisplaybreaks
\begin{align}
\mathcal{E}&= \max_{|\Lambda| = L} \EE_{\cC} \max_{z \in \F_q^k} \sum_{x \in \Lambda} \agr(c(x), z) \notag \\
&\leq \max_{|\Lambda| = L} \EE_{\cC} \max_{z \in \F_q^k}\inparen{ \frac{nL}{q} + \frac{ nL}{2\eps} \inparen{ 1 + \eps^2 }\inparen{1 - \frac{1}{q}} - \frac{n}{2L\eps}\sum_{x \neq y \in \Lambda} d(c(x),c(y)) } \notag \\
&= \max_{|\Lambda| = L}\inparen{\frac{nL}{q} + \frac{ nL}{2\eps} \inparen{ 1 + \eps^2 }\inparen{1 - \frac{1}{q}} - \frac{n}{2L\eps}\sum_{x \neq y \in \Lambda} \EE_{\cC} d(c(x),c(y)) }\notag \\
\label{step:jb-4}
&\leq \frac{nL}{q} + \frac{ nL}{2\eps} \inparen{ 1 + \eps^2 }\inparen{1 - \frac{1}{q}} - \frac{n(L-1)\inparen{1 - \frac{1}{q} - \frac{\eps^2}{2} }}{2\eps}  \\
&= \frac{nL}{q} + \frac{ nL \eps}{2} \inparen{ \frac{3}{2} - \frac{1}{q} } + \frac{ n \inparen{ 1 - \frac{1}{q} - \frac{\eps^2}{2} } }{2 \eps} \notag \\
&\leq \frac{nL}{q} + \frac{ 3nL\eps}{4} + \frac{ n}{2\eps }\notag \\
\label{step:jb-bound}
&\le nL\inparen{\frac{1}{q}+\eps}.
\end{align}
}
In the above, \eqref{step:jb-4} follows from the fact that the original code had (relative) distance $1 - 1/q - \eps^2/2$ and that in the construction of $\mathcal{C}$ from $\mathcal{C}'$, pairwise Hamming distances are preserved in expectation. Finally, \eqref{step:jb-bound} follows from the assumption that $L\ge 2/\eps^2$.

Recall from the statement of Theorem \ref{thm:mainthm} that we have defined
\[ Y = \Cthm L \log(N) \log^5(L),\]
so the assumption on $n$ implies that
\[ Y \leq nL\min\{\eps, q\eps^2\}. \]
Suppose that $q\eps \leq 1$, so that $Y \leq nLq\eps^2$.  
Plugging this along with \eqref{step:jb-bound} into Theorem \ref{thm:mainthm}, we obtain
\begin{align}
\EE_\mathcal{C} \max_{z\in \F_q^n} \max_{\Lambda \subset \F_q^k, |\Lambda| = L} \sum_{x \in \Lambda} \agr(c(x),z) &\leq 
\mathcal{E} + Y + \sqrt{ \mathcal{E} Y } \notag\\
&\leq nL\inparen{\frac{1}{q} + \eps} + nLq\eps^2 + nL \sqrt{ q\eps^2\inparen{\frac{1}{q} + \eps } } \notag\\
&= nL\inparen{\frac{1}{q} + \eps\inparen{ 1 + q\eps + \sqrt{1 + q\eps} }}\notag\\
&\leq nL \inparen{ \frac{1}{q} + \eps \inparen{ 2 + \sqrt{2} } }, \notag
\end{align}
using the assumption that $q\eps \leq 1$ in the final line.
Thus, Proposition \ref{prop:dec} implies that $\mathcal{C}$ is 
$\inparen{1 - \nicefrac{1}{q} - (2 + \sqrt{2})\eps, 2/\eps^2}$-list-decodable.

On the other hand, suppose that $q\eps \geq 1$, so that $Y \leq nL\eps$.  Then following the same outline, we have
\begin{align}
\EE_\mathcal{C} \max_{z\in \F_q^n} \max_{\Lambda \subset \F_q^k, |\Lambda| = L} \sum_{x \in \Lambda} \agr(c(x),z) &\leq 
\mathcal{E} + Y + \sqrt{ \mathcal{E} Y } \notag\\
&\leq nL\inparen{\frac{1}{q} + \eps} + nL\eps + nL \sqrt{ \eps\inparen{\frac{1}{q} + \eps } } \notag\\
&= nL\inparen{\frac{1}{q} + \eps\inparen{ 2 + \sqrt{\frac{1}{q\eps} + 1} }}\notag\\
&\leq nL \inparen{ \frac{1}{q} + \eps \inparen{ 2 + \sqrt{2} } }, \notag
\end{align}
using the assumption that $q\eps \geq 1$ in the final line.  Thus, in this case as well, $\mathcal{C}$ is 
$\inparen{1 - \nicefrac{1}{q} - (2 + \sqrt{2})\eps, 2/\eps^2}$-list-decodable.

This completes the proof of Corollary \ref{cor:smallq}.
\end{proof}

\begin{proof}[Proof of Corollary \ref{cor:largeq}]
As with Corollary \ref{cor:smallq}, we need an average-radius version of the Johnson bound.  In this case, we 
follow a proof of the Johnson bound from~\cite{MS77}, which gives a better dependence on $\eps$ in the list size when $q$ is large.
For completeness, our average-radius version of the proof is given in Appendix \ref{sapp:avgjb}, Theorem \ref{thm:jb2}.

We proceed with the proof of Corollary \ref{cor:largeq}.  
By Theorem \ref{thm:jb2}, for any $z \in \F_q^n$ and for all $\Lambda \subset \F_q^k$ of size $L$, 
\begin{equation}\label{eq:jb2}
[\sum_{x\in\Lambda} \agr(c(x),z) \le \frac{1}{2}\inparen{n+\sqrt{n^2+ 4n^2L(L-1)-4n^2\sum_{x\neq y\in\Lambda} d(c(x),c(y))}}.
\end{equation}

By Theorem \ref{thm:mainthm}, it suffices to control $\mathcal{E}$. 
Since the right hand side above does not depend on $z$,
{\allowdisplaybreaks
\begin{align}
\mathcal{E}&= \max_{|\Lambda| = L} \EE_{\cC} \max_{z \in \F_q^k} \sum_{x \in \Lambda} \agr(c(x), z) \notag \\
\label{eq:jb2-step1}
&\le \max_{|\Lambda| = L} \EE_{\cC} \max_{z \in \F_q^k} \inparen{\frac{1}{2}\inparen{n+\sqrt{n^2+ 4n^2L(L-1)-4n^2\sum_{x\neq y\in\Lambda} d(c(x),c(y))}}}\\
\label{eq:jb2-step2}
&\le \max_{|\Lambda| = L} \frac{1}{2}\inparen{n+\sqrt{n^2+ 4n^2L(L-1)-4n^2\sum_{x\neq y\in\Lambda} \EE_{\cC} d(c(x),c(y))}}\\
\label{eq:jb2-step3}
&\le \frac{1}{2}\inparen{n+\sqrt{n^2+ 4n^2L(L-1)-4n^2\sum_{x\neq y\in\Lambda} (1-\eps^2)}}\\
&\le \frac{1}{2}\inparen{n+\sqrt{n^2+ 4n^2L(L-1)\eps^2}} \notag\\
&< \frac{1}{2}\inparen{n+\sqrt{n^2+ 4n^2L^2\eps^2}} \notag\\
\label{eq:jb2-bound}
&\le 2nL\eps.
\end{align}
}
In the above, \eqref{eq:jb2-step1} follows from \eqref{eq:jb2}. \eqref{eq:jb2-step2} follows from Jensen's inequality. \eqref{eq:jb2-step3} 
follows from the fact that the original code had (relative) distance $1 - \eps^2$ and that in the construction of $\mathcal{C}$ from $\mathcal{C}'$, pairwise Hamming distances are preserved in expectation. Finally, \eqref{eq:jb2-bound} follows from the assumption that $L\ge 1/\eps$. 

Now, Theorem \ref{thm:mainthm} implies that
\begin{align*}
\EE_\mathcal{C} \max_{z\in \F_q^n} \max_{\Lambda \subset \F_q^k, |\Lambda| = L} \sum_{x \in \Lambda} \agr(c(x),z) &\leq 
\mathcal{E} + Y + \sqrt{ \mathcal{E} Y } \\
&\leq 2\inparen{ \mathcal{E} + Y } \\
&\leq 2\inparen{ 2nL\eps + Y }\\
&\leq 5nL\eps
\end{align*}
where as before
\[ Y = \Cthm L\log(N) \log^5(L)\]
and where we used the choice of $n$ in the final line.
Choose $\eps' = 5\eps$, so that whenever $5\eps > 1/q$, Proposition \ref{prop:dec} applies and completes the proof. 
Because we have chosen $\eps > 1/\sqrt{q}$ (which is necessary in order for $\mathcal{C}'$ to have distance $1 - \eps^2$), 
the condition that $5\eps > 1/q$ holds for sufficiently small $\eps$.
\end{proof}

\section{Background on Gaussian random variables}\label{app:gauss}
In this appendix, we record a few useful facts about Gaussian random variables which we use in the body of the paper.  
Next, we justify the claim \eqref{eq:expgauss} from Section \ref{sec:prooffromlemma}.  
Finally, we justify the claim \eqref{eq:symmetrize} from Section \ref{sec:redux}.
These facts are standard, and can be found, for example, in~\cite{lt}.

\subsection{Some facts about Gaussians}

A \textbf{gaussian random variable} $g \sim N(0,\sigma^2)$ with variance $\sigma^2$ has a probability density function
\[ f(t) = \frac{1}{\sqrt{2\pi\sigma^2}} \exp( -t^2/2\sigma^2 ). \]
The cumulative distribution function, 
\[ \PR{ g > t } = \frac{1}{\sqrt{2\pi\sigma^2}} \int_{u=t}^\infty \exp(-u^2/2\sigma^2)\, du \]
obeys the estimate 
\begin{equation}\label{eq:tail}
 \PR{ g > t} \leq \frac{\sigma}{t} \cdot \frac{1}{\sqrt{2\pi}} \exp( -t^2/ 2\sigma^2) 
\end{equation}
for all $t > 0$.  Indeed, because on the domain $u \geq t$, $(u/t) \geq 1$, we have
\begin{equation}\label{eq:tailderiv}
 \frac{1}{\sqrt{2\pi\sigma^2} }\int_{u=t}^\infty \exp\inparen{\frac{-u^2}{2\sigma^2} }\, du 
\leq \frac{1}{\sqrt{2\pi\sigma^2}}  \int_{u=t}^\infty \frac{u}{t} \exp\inparen{ \frac{-u^2}{2\sigma^2}}\,du = \frac{ \sigma }{t \sqrt{2\pi} }\exp\inparen{\frac{-t^2}{2\sigma^2}}. 
\end{equation}

Linear combinations of Gaussian random variables are again Gaussian.  
\begin{fact}\label{fact:2stab} Let $g_1, \ldots, g_n$ be Gaussian random variables with variances $\sigma_1^2,\ldots, \sigma_n^2$.  Then the random variable $\sum_i a_i g_i$
is again a Gaussian random variable, with variance
$ \sum a_i^2 \sigma_i^2.$
\end{fact}

In the body of the paper, we use the following bound on the expected value of the maximum of $n$ Gaussian random variables with variances $\sigma_1^2, \ldots, \sigma_n^2 \leq \sigma^2$.  
\begin{proposition}\label{prop:maxofgauss}
Let $g_i \sim N(0,\sigma_i^2)$, for $i = 1,\ldots,n$, and suppose that $\max_i \sigma_i \leq \sigma$.  Then 
\[ \EE \max_{i \in [n]} |g_i| \leq \sigma \sqrt{2\ln(n)}\cdot (1 + o(1)).\]
\end{proposition}
\begin{proof}
We have
\begin{align*}
 \EE \max_{i \in [n]} |g_i| &= \int_{u=0}^\infty \PR{ \max_{i \in [n]} |g_i| > u}\,du \\
&\leq A + \frac{2}{\sqrt{2\pi } } \int_{u=A}^\infty n\exp\inparen{ \frac{-u^2}{ 2\sigma^2} }\,du
\end{align*}
for any $A\ge \sigma$ (which we will choose shortly). In the above inequality, we have used \eqref{eq:tail} (with the fact that $A\ge \sigma$) and the fact that for every $i$, $\PR{|g_i|>u} =2\PR{g_i>u}$. We may estimate the integral using \eqref{eq:tailderiv}, so
\[ \frac{2}{\sqrt{2\pi} } \int_{u=A}^\infty \exp\inparen{\frac{-u^2}{2\sigma^2} }\,du \leq \frac{ 2\sigma^2 }{A \sqrt{2\pi}} \exp\inparen{ -\frac{A^2}{2\sigma^2}}. \]
Choosing $A = \sigma\sqrt{2 \ln(n) }$, we get
\[ \EE \max_{i \in [n]} |g_i| \leq \sigma \sqrt{2\ln(n)} + \frac{\sigma}{\sqrt{\pi \ln(n)}}. \]
\end{proof}

\subsection{Justification of \eqref{eq:expgauss}}\label{app:justexp}
We use a computation similar to that in the proof of Proposition \ref{prop:maxofgauss} to justify
\eqref{eq:expgauss}, which states that
\begin{equation*}
 \EE \max_{(I,\Lambda) \in \mathcal{S}_0} S(I,\Lambda) \leq 2\sum_{t=1}^{\tm - 1} a_t  =: 2A.
\end{equation*}
Recall that we had shown that 
\[ \PR{ \max_{(I,\Lambda) \in \mathcal{S}_0} S(I,\Lambda) > u \cdot \sum_{t=0}^{\tm - 1} a_t } \leq  \sum_{t=0}^{\tm-1} N_tN_{t+1} \exp\inparen{ - \frac{u^2 \cdot a_t^2}{\delta_t^2}},\]
and that we had chosen
\[ a_t =  \sqrt{2\ln\inparen{N_tN_{t+1}}}\,\delta_t. \]
Now \eqref{eq:expgauss} follows from a computation similar to the proof of Proposition \ref{prop:maxofgauss}.  Indeed, we have
\begin{align*}
\EE \max_{ (I,\Lambda) \in \mathcal{S}_0} S(I,\Lambda) &= \int_{u=0}^\infty \PR{ \max_{(I,\Lambda)} S(I,\Lambda) > u}\,du\\
&\leq A+\int_{u=A}^\infty \sum_{t=0}^{\tm-1} N_tN_{t+1} \exp\inparen{ \frac{-u^2 \cdot a_t^2 }{\delta_t^2 A^2 }}\,du\\
&= A+ \int_{u=A}^\infty \sum_{t=0}^{\tm -1} N_tN_{t+1} \exp\inparen{ \frac{-2 u^2 \ln\inparen{N_tN_{t+1}}}{ A^2 } } \,du\\
&\leq A + \sum_{t=0}^{\tm -1}N_t N_{t+1} \int_{u=A}^\infty \exp\inparen{ \frac{-2 u^2 \ln\inparen{N_tN_{t+1}}}{ A^2 } }\,du.
\end{align*}
Repeating the trick \eqref{eq:tailderiv}, we estimate
\[ \int_{u=A}^\infty \exp\inparen{ \frac{-2 u^2 \ln\inparen{ N_t N_{t+1} }}{A^2} }
\leq \frac{A}{4  \ln\inparen{ N_t N_{t+1}}} \exp\inparen{ - 2 \ln\inparen{ N_t N_{t+1}}}
\leq \frac{A}{4 N_t^2N_{t+1}^2 }.\]
Plugging this in, we get
\[ \EE \max_{ (I,\Lambda) \in \mathcal{S}_0} S(I,\Lambda) \leq A\inparen{ 1 + \frac{1}{4} \sum_{t=0}^{\tm - 1} \frac{1}{N_t N_{t+1}} } \leq 2A. \]
In the last inequality, we used the definition of $N_t = \Csize {N \choose eL/2^t}{N \choose e L/2^{t+1}}$ if $t\ge 1$ and $N_0=\binom{N}{L}$. In particular, we have used the fact that $N_t\ge 2$ for our setting of parameters.

\subsection{Justification of \eqref{eq:symmetrize}}\label{app:justsymm}
Finally, we justify \eqref{eq:symmetrize}, which read
\[ \EE_\cC \max_{|\Lambda| = L} \inabs{ \sum_{j \in [n]} \inparen{\pl_j(\Lambda) - \EE_\cC \pl_j(\Lambda)}}  \leq \Ccmp\, \EE_\cC \EE_g \max_{|\Lambda| = L} \inabs{ \sum_{j\in [n]} g_j \pl_j(\Lambda) }.\]
Recall that the $\pl_j(\Lambda)$ are independent random variables.  We proceed in two steps; first, a symmetrization argument will introduce Rademacher random variables\footnote{That is, random variables which take the values $+1$ and $-1$ with probability $1/2$ each.} $\xi_i$, and next a comparison argument will replace these with Gaussian random variables.  Both steps are standard, 
 and more general versions are given in~\cite{lt} as Lemma 6.3 and Equation (4.8), respectively.  Here, we state and prove simplified versions for our needs.

We begin by symmetrizing the left hand side of \eqref{eq:symmetrize}.
\begin{lemma}\label{lem:symmetrize}
With $\pl_j(\Lambda)$ as above, 
\[ \EE \max_{|\Lambda| = L} \inabs{ \sum_{j \in [n]} \pl_j(\Lambda) - \EE \pl_j(\Lambda) } \leq 2 \EE \max_{|\Lambda| = L} \inabs{ \sum_{j \in [n]} \xi_j \pl_j(\Lambda) }, \]
where the $\xi_i$ are independent Rademacher random variables.
\end{lemma}
\begin{proof}
Let $\cC'$ be an independent copy of $\mathcal{C}$, and let
$\pl_j'(\Lambda)$ denote an independent copy of $\pl_j(\Lambda)$.  Then, 
\begin{align*}
\EE_{\cC} \max_{\Lambda} \inabs{ \sum_{j \in [n]} \pl_j(\Lambda) - \EE_{\cC} \pl_j(\Lambda)}  
&= \EE_{\cC} \max_{\Lambda} \inabs{ \sum_{j \in [n]} \pl_j(\Lambda) - \EE_{\cC} \pl_j(\Lambda) - \EE_{\cC'} \left[ \pl_j'(\Lambda) - \EE_{\cC'} \pl_j'(\Lambda) \right] }\\
&\leq \EE_{\cC}\EE_{\cC'} \max_{\Lambda} \inabs{ \sum_{j \in [n]} \pl_j(\Lambda) - \pl'_j(\Lambda) } \qquad \text{ \begin{tabular}{l} by Jensen's inequality,\\ and because $\EE_{\cC}\pl_j(\Lambda) = \EE_{\cC'} \pl_j'(\Lambda)$ \end{tabular} } \\
&= \EE_\xi \EE_{\cC} \EE_{\cC'} \max_{\Lambda} \inabs{ \sum_{j \in [n]} \xi_j (\pl_j(\Lambda) - \pl'_j(\Lambda) ) } \qquad \text{ \begin{tabular}{l} by independence, and \\ the fact that $\pl_j(\Lambda)$ and $\pl'_j(\Lambda)$ \\ are identically distributed \end{tabular}}\\
&\leq 2\EE_\xi \EE_{\cC} \max_{\Lambda} \inabs{ \sum_{j \in [n]} \xi_j \pl_j(\Lambda) } \qquad \text{by the triangle inequality.}
\end{align*}
\end{proof}

Next, we replace the Rademacher random variables $\xi_j$ with Gaussian random variables $g_j$ using a comparison argument.
\begin{lemma}\label{lem:compare}
Condition on the choice of $\mathcal{C}$, and let $\pl_j(\Lambda)$ be as above.
Let $\xi_1, \ldots, \xi_n$ be independent Rademacher random variables, and let $g_1,\ldots, g_n$ be independent standard normal random variables.
Then
\[ \EE_\xi \max_{|\Lambda| = L} \inabs{ \sum_{j \in [n]} \xi_j \pl_j(\Lambda) }\leq \sqrt{ \frac{\pi}{2}} \EE_g \max_{|\Lambda| = L} \inabs{ \sum_{j \in [n]} g_j \pl_j(\Lambda)}. \]
\end{lemma}
\begin{proof}
We have
\begin{align*}
\EE_g \max_{\Lambda} \inabs{ \sum_{j \in [n]} g_j \pl_j(\Lambda) }
&= \EE_g \EE_\xi \max_{\Lambda} \inabs{ \sum_{j \in [n]} \xi_j |g_j| \pl_j(\Lambda) }\\
&\geq \EE_{\xi} \max_{\Lambda} \inabs{ \sum_{j \in [n]} \xi_j \EE_g |g_j| \pl_j(\Lambda) } \qquad \text{ by Jensen's inequality } \\
&= \EE_{\xi} \max_{\Lambda} \inabs{ \sum_{j \in [n]} \xi_j \sqrt{ \frac{2}{\pi} } \pl_j(\Lambda) }. 
\end{align*}
Above, we used the fact that for a standard normal random variable $g_j$, $\EE |g_j| = \sqrt{ 2/\pi}$.  
\end{proof}

Together, Lemma \ref{lem:symmetrize} and Lemma \ref{lem:compare} imply \eqref{eq:symmetrize}.

\end{document}